\documentclass[english]{IEEEtran}
\usepackage{subfigure}
\usepackage[T1]{fontenc}
\usepackage[latin9]{inputenc}
\usepackage{amsthm}
\usepackage{amsmath}
\usepackage{amssymb}
\usepackage{graphicx}
\usepackage[english]{babel}
\usepackage{epstopdf}
\usepackage{pgfplots}
\usepackage{tikz}
\usepackage{subfig}

\pgfplotsset{compat=newest}
\pgfplotsset{plot coordinates/math parser=false}
\newlength\figureheight
\newlength\figurewidth

\makeatletter

\theoremstyle{plain}
\newtheorem{thm}{\protect\theoremname}
  \theoremstyle{remark}
  \newtheorem{rem}{\protect\remarkname}

\usepackage{amsfonts}
\usepackage{cite}
\usepackage{array}
\usepackage{algorithm}
\usepackage{algorithmic}
\usepackage{subfigure}
\usepackage{color}
\DeclareMathOperator{\rank}{rank}
\DeclareMathOperator{\tr}{tr}
\DeclareMathOperator{\diag}{diag}

\DeclareMathOperator*{\mini}{min}
\DeclareMathOperator*{\maxi}{max}
\DeclareMathOperator*{\st}{s.t.}

\makeatother

\makeatother
\providecommand{\remarkname}{Remark}
\providecommand{\theoremname}{Theorem}

\begin{document}
\title{An Efficient Precoder Design for Multiuser MIMO Cognitive Radio Networks with Interference Constraints
}
\author{
\IEEEauthorblockN{Van-Dinh~Nguyen,~Le-Nam~Tran, \textit{Member, IEEE},~Trung~Q.~Duong, \textit{Senior Member, IEEE,}\\~Oh-Soon~Shin, \textit{Member, IEEE},~and~Ronan~Farrell, \textit{Member, IEEE}}
\thanks{V.-D.~Nguyen and O.-S.~Shin are with the School of Electronic Engineering, Soongsil University, Seoul 06978, Korea (e-mail: \{nguyenvandinh, osshin\}@ssu.ac.kr).}
\thanks{T.~Q.~Duong is with the School of Electronics, Electrical Engineering and Computer Science, Queen's University Belfast, Belfast BT7 1NN, United Kingdom (e-mail: trung.q.duong@qub.ac.uk).}
\thanks{L.-N.~Tran and R.~Farrell are with the Department of Electronic Engineering, Maynooth University, Ireland (e-mail: \{ltran, rfarrell\}@eeng.nuim.ie).
}
}
\maketitle
\begin{abstract}
We consider a linear precoder design for an underlay cognitive radio multiple-input multiple-output broadcast channel, where the secondary system consisting of a secondary base-station (BS) and a group of secondary users (SUs) is allowed to share the same spectrum with the primary system. All the transceivers are equipped with multiple antennas, each of which has its own maximum power constraint. Assuming zero-forcing method to eliminate the multiuser interference, we study the sum rate maximization problem for the secondary system subject to both per-antenna power constraints at the secondary BS and the interference power constraints at the primary users. The problem of interest differs from the ones studied previously that often assumed a sum power constraint and/or single antenna employed at either both the primary and secondary receivers or the primary receivers.
To develop an efficient numerical algorithm, we first invoke the rank relaxation method to transform the considered problem into  a convex-concave problem based on a downlink-uplink result. We then propose a barrier interior-point method to solve the resulting saddle point problem.
In particular, in each iteration of the proposed method we find the Newton step by solving a system of discrete-time Sylvester equations, which help reduce the complexity significantly, compared to the conventional method.  Simulation results are provided to demonstrate  fast convergence and effectiveness of the proposed algorithm.
\end{abstract}
\begin{IEEEkeywords}MIMO, broadcast channel, beamforming, cognitive radio, zero-forcing, sum rate maximization.
\end{IEEEkeywords}

\section{Introduction}
Radio frequency spectrum has recently become a scarce and expensive wireless resource due to 	the ever increasing demand of multimedia services. Nevertheless,  it has been reported in \cite{SpectrumPolicy} that the majority of  licensed users are idle at any  given time and location. To significantly improve spectrum utilization, cognitive radio (CR) is widely considered as a promising approach. There are two famous CR models, namely opportunistic spectrum access model and spectrum sharing model. In the former, the secondary users (SUs, unlicensed users) are allowed to use the frequency bands of the primary users (PUs, licensed users) only when these bands are not occupied \cite{Mitola:99,Haykin:05,Gan:VT:12,Peh:VT:09,He:VT:10}.  In the latter, the PUs have prioritized access to the available radio spectrum, while the SUs have restricted access and  need to avoid causing detrimental interference to the primary receivers \cite{Wang:VT:14,Hua:VT:14, Zhang:JBPA:08}. Towards this end, several powerful techniques such as spectrum sensing and beamforming design, have been used to protect the PUs from the interference from the SUs and to meet quality-of-service (QoS) requirements \cite{Khan:VT:12, Sun:VT:15, Gharavol:VT:10,Zhang:08}.

It is well-known that  transmit beamforming has enormous potential to improve the capacity of wireless communication systems without requiring extra bandwidth or transmit power. In fact, a popular method to control the interference that has been widely used in the literature is based on zero-forcing (ZF) technique, which places null spaces at the beamforming vector of each co-channel receiver. Specifically, successive zero-forcing dirty paper coding (SZF-DPC) was introduced in \cite{Caire:03} for single-antenna receivers and was extended  to multiple-antenna receivers in \cite{Dabbagh:07}. Inspired by these two works,  the authors in \cite{Tran:MISO:13,Tran:MIMO:13} proposed  efficient numerical methods based on a barrier method to solve the throughput maximization problem. These approaches were shown to have a superior convergence behavior compared to a two-stage iterative method based on the dual subgradient method \cite{Zhang:BDCooperative:2010}.

The throughput maximization of a CR  network was optimally solved in \cite{Zhang:08}, which 	considered a single secondary multiple-input multiple-output (MIMO)/multiple-input single-output (MISO) link under the constraint of opportunistic spectrum sharing. Multiple antennas are exploited at the secondary transmitter to optimally trade off between attaining the maximal throughput and meeting the interference threshold at the primary receivers.  The throughput of the SU link was studied with instantaneous or average interference-power constraint at a single PU in \cite{Ghasemi:07}. When multiple SUs access a single-frequency which is licensed to one PU,  a QoS constraint for each secondary link is solved efficiently  using a geometric programming method   under the interference-power constraint at some measured point \cite{Huang:06, Xing:07}.  The authors in \cite{Tajer:BRA:10} considered maximizing the smallest weighted rate and proposed distributed algorithms for optimal beamforming and rate allocation.  The channel capacity of CR networks with two users was presented in \cite{Jafar:07}, where  the cognitive user is assumed to know the message
of the primary user non-causally prior to transmissions. The impact of average interference power  and peak interference power constraints on the  ergodic capacity of CR networks was compared in \cite{Zhang:WC:09}. In \cite{Stotas:WC:12}, an enhanced spectrum sensing scheme was proposed  to improve the optimal power allocation strategy in CR networks. In parallel lines, beamforming design technique was applied to control the total amount of interference caused by the secondary transmitter, thereby enhancing the channel capacity of CR MIMO networks \cite{Rezaei:VT:15,Luan:VT:12}.

The sum rate (SR) maximization problem of CR has been a classical one  in wireless system design \cite{Zhang:WSR:09,Gallo:WSR:11,KIM:ORA:11,Jitvan:BPC:09}.  In particular,  the authors in \cite{Zhang:WSR:09} considered the weighted sum rate (WSR) maximization problem for CR multiple-SUs MIMO broadcast channel (BC) under the sum power constraint and the interference power constraints. The problem is shown to be a nonconvex problem but can be transformed into an equivalent convex CR MIMO multiple access channel (MAC) problem via the general BC-MAC duality result.  The problem of WSR maximization for multiple MISO SUs in the presence of multiple primary receivers was addressed in \cite{Gallo:WSR:11}, where  an iterative algorithm based on the subgradient method was proposed to determine the beamforming vectors. For the MIMO ad hoc CR network   considered in  \cite{KIM:ORA:11}, a semidistributed algorithm and a centralized algorithm based on geometric programming and network duality were introduced under the interference constraint at the primary receivers in order to obtain a locally optimal linear precoder for the non-convex WSR maximization problem. In \cite{Jitvan:BPC:09}, the authors investigated a CR multiple-antenna two-way relay network, where  optimal relay beamforming and suboptimal beamforming vectors were obtained based on the subspace projection and power control algorithms that are proposed to satisfy the power budget as well as the interference power constraints. 	 However, all of these algorithms generally have high computational complexity  and slow convergence.

Very recently, the SR maximization  for CR MISO-BC with a very large number of antennas at the secondary transmitter was studied in \cite{He:SRM:14}. For the general case, a low-complexity suboptimal beamforming scheme based on partially-projected regularized ZF beamforming was applied to obtain a closed-form beamformer, and then  deterministic approximations for large system analysis were also derived. The authors in \cite{Zheng:MSR:14, Lai:BDA:15} studied a WSR maximization problem with a collection of multiple PUs and SUs  sharing a common frequency-flat fading channel.
 Under the interference power and transmit power constraint,   a reformulation-relaxation technique was proposed to solve a convex optimization problem over a closed bounded convex set in the rate domain to bound the optimal value, and a branch-and-bound algorithm is used to find a globally optimal solution to the WSR maximization problem. Most of prior research on beamforming design for CR networks assumed that  single antenna is employed at either both the primary  and secondary receivers or the primary receiver. To the best of  our knowledge, the SR maximization problem in more general cases for CR MIMO-BC networks in the presence of multiple SUs and  PUs is still an open research and has not been investigated previously.

In this paper we consider a CR network where the secondary system consisting  of a multiantenna BS sends data to multiple users on its downlink channel while satisfying the interference limit that it may cause to the legitimate users of the primary system. We note that the capacity of the downlink transmission of the secondary system can be achieved through DPC \cite{Weingarten:CSG:06}. However, DPC
is difficult to implement in practice since it is a nonlinear precoding technique with high implementation complexity. In this paper we  adopt ZF precoding which has been widely used in the literature thanks to its simplicity and effectiveness at the secondary base station (BS), and consider the SR maximization problem subject to per-antenna power constraints (PAPC) at secondary BS and the interference power constraint at PUs.
By a standard rank relaxation method, the precoder design problem can be cast as a semidefinite program which can be solved by generic conic solvers such as SeDuMi \cite{Sturm} or SDPT3 \cite{Toh}. However, we do not follow such approach for the following two reasons. First, it provides little useful insights into the structure of the optimal precoder. Second, its computational complexity is generally very high since specifications of the considered problem are not exploited.
In this paper,  our contributions include the followings
\begin{itemize}
  \item We reformulate the SR maximization problem for the considered CR MIMO system subject to PAPCs and interference constraints into a concave-convex (also known as minimax) program by extending  a BC-MAC duality result. To do this we first apply a standard rank relaxation method and then prove that the relaxation is tight. In particular, we  show that the PAPCs and interference constraints will become equality constraints in the derived convex-concave program.
  \item We customize the barrier method to find a saddle point (i.e., an optimal solution) of the convex-concave program. The proposed algorithm is  basically an iterative Newton-type method which is customized to exploit the special features of the design problem. Explicitly, in each iteration to find the Newton step, we arrive at a system of Sylvester equations which  is less complex to solve, compared to a generic method based on solving a system of linear equations.
  \item We provide extensive numerical results to justify the novelty of the algorithm and compare its performance with suboptimal solutions. In particular,  the numerical results demonstrate a superlinear convergence rate of the proposed algorithm and superior performance on achievable sum rate over the suboptimal solutions.
\end{itemize}

  The rest of the paper is organized as follows. System model and the formulation of the SR maximization problem are described in Section \ref{Systemmodel}. In Section \ref{Proposedalgorithm}, we present the proposed algorithm to solve the considered problem. Numerical results are provided in Section \ref{Numericalresults}, and Section \ref{Conclusion} concludes the paper.

\emph{Notations}: Bold lower and upper case letters represent vectors and matrices, respectively; $\mathbf{H}^{H}$ and $\mathbf{H}^{T}$ are Hermitian and normal transpose of $\mathbf{H}$, respectively; $\tr(\mathbf{H})$ and $|\mathbf{H}|$ are the trace and the determinant of $\mathbf{H}$, respectively; $\mathbf{I}_N$ represents an $N\times N$ identity matrix. $[\mathbf{x}]_i$ is the $i$th entry of vector $\mathbf{x}$. $[\mathbf{H}]_{i,j}$ is the entry at the $i$th row and $j$th column of $\mathbf{H}$. $\mathbf{e}_{n}$ is the $n$th unit vector (all entries are zero except for the  $n$th element which is $1$). $\diag(\mathbf{x})$, where $\mathbf{x}$ is a vector, denotes a diagonal matrix with diagonal elements of $\mathbf{x}$. The notation $\mathbf{X}\succeq\mathbf{0}$ represents that $\mathbf{X}$ is a positive semidefinite matrix. $\mathcal{N}(\mathbf{H})$ denotes the null space of $\mathbf{\mathbf{H}}$ and  $\lambda_{\max}(\mathbf{X}) $ is the maximum eigenvalue of $\mathbf{X}$.

\section{System Model and Problem Formulation}\label{Systemmodel}

\begin{figure}[t]
\centering
\includegraphics[trim=0.00cm 0cm 0.0cm 0.0cm,width=0.35\textwidth]{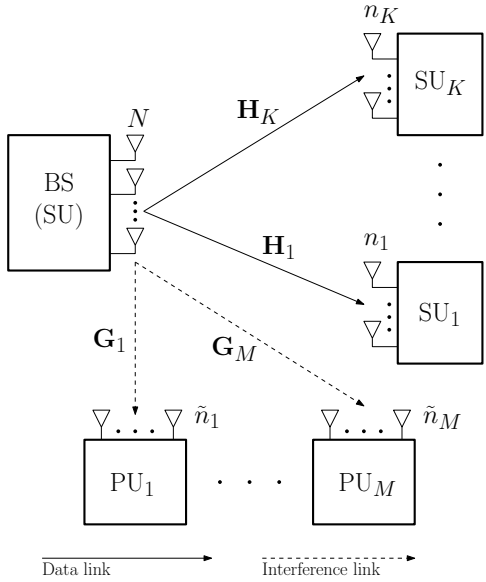}
\caption{A cognitive network model with multiple SUs and PUs.}
\label{fig:systemmodel}
\end{figure}
We consider a cognitive transmission scenario where an $N$-antenna secondary BS sends data to $K$   SUs in the presence of $M$  PUs as shown in Fig. \ref{fig:systemmodel}.  The secondary system is allowed to share the same spectrum  band licensed to the primary system. In an underlay cognitive network, the secondary BS must adapt its transmit power to satisfy an interference power constraint at the $m$th PU which is denoted by $I_m$, for $m = 1, 2,\cdots, M$. The numbers of receive antennas at the $k$th SU and the $m$th PU are denoted by $n_{k}$ and $\tilde{n}_m$, respectively, and their channel matrices (to the secondary BS) are represented by $\mathbf{H}_{k}\in\mathbb{C}^{n_{k}\times N},\; \text{for}\; k=1,\cdots, K$, and $\mathbf{G}_{m}\in\mathbb{C}^{\tilde{n}_{m}\times N},\; \text{for}\; m=1,\cdots, M$, respectively. We  assume that $\mathbf{H}_{k},\;\forall k$ and $\mathbf{G}_{m},\;\forall m$ remains constant during a transmission block and change independently from one block
to another. We further assume that the channel matrices can be perfectly known at the secondary BS, using a genie added feedback.  Though this assumption is quite ideal,  it has been considered in  \cite{He:SRM:14,Gallo:WSR:11,Zhang:WSR:09} to study their problems of interest in the context of CR networks. In reality, perfect channel estimation is hardly achieved and thus the results  obtained in this paper may act as an upper bound on the SR performance for the secondary transmission in an underlay CR network.

In the considered system, linear precoding is  employed at the secondary BS to transmit data to the SUs. Specifically,  the vector of transmitted symbols of the $k$th SU, denoted by  ${\rm {\bf x}}_{k}$, is multiplied by the precoder $\mathbf{T}_{k}\in\mathbb{C}^{N\times L_{k}}$, $L_k\leq\min(N, n_k)$, before being transmitted.  We assume that $E[\mathbf{x}_{k}\mathbf{x}_{k}^{H}]=\mathbf{I}$. In this way, the complex baseband transmitted signal at the secondary BS can be expressed as
\begin{equation}\label{eq:basebandsinal}
\mathbf{x} = \sum_{k=1}^{K}\mathbf{T}_k\mathbf{x}_k
\end{equation}
and the received signal at the $k$th SU is given by
\begin{equation}
\mathbf{y}_{k}=\mathbf{H}_{k}\mathbf{T}_{k}\mathbf{x}_{k}+\sum_{j\neq k}\mathbf{H}_{k}\mathbf{T}_{j}\mathbf{x}_{j}+\mathbf{z}_{k}\label{eq:signal:model}
\end{equation}
where $\mathbf{z}_{k}\in\mathbb{C}^{n_{k}\times1}$ is the background noise with distribution $\mathcal{CN}(\mathbf{0},\mathbf{I}_{n_{k}})$.  According to ZF precoding, we need to design $\mathbf{T}_{j}$ such that $\mathbf{H}_{k}\mathbf{T}_{j}=\mathbf{0}$
for all $j\neq k$. Consequently, the problem of SR maximization for cognitive transmission is mathematically formulated as
\begin{IEEEeqnarray}{rCl} \label{eq:BD:Sumrate:PAPC}
\underset{\{\mathbf{T}_{k}\}}{\maxi}\quad &  & \sum_{k=1}^{K}\log|\mathbf{I}+\mathbf{H}_{k}\mathbf{T}_{k}\mathbf{T}_{k}^{H}\mathbf{H}_{k}^{H}| \IEEEyesnumber \IEEEyessubnumber\label{eq:obj:2}\\
\st\quad &  & \mathbf{H}_{k}\mathbf{T}_{j}=\mathbf{0},\:\forall j\neq k\IEEEyessubnumber\label{eq:zfcondi}\\
 &  & \sum_{k=1}^{K}[\mathbf{T}_{k}\mathbf{T}_{k}^{H}]_{n,n}\leq P_{n},\;\forall n\in\mathcal{N}\IEEEyessubnumber\label{eq:PAPC}\\
 &  & \sum_{k=1}^{K}\tr(\mathbf{G}_{m}\mathbf{T}_{k}\mathbf{T}_{k}^{H}\mathbf{G}_{m}^{H})\leq I_{m},\;\forall m\in\mathcal{M}\IEEEyessubnumber\label{eq:threshold}
\end{IEEEeqnarray}
where $\mathcal{N}\triangleq\{1, 2,\cdots, N\}$ and $\mathcal{M}\triangleq\{1, 2,\cdots, M\}$. Here, 	the constraint in \eqref{eq:PAPC} is the  power constraints for the $n$th antenna at the secondary BS. We remark that an antenna  is often equipped with its own power amplifier (PA). Thus, we may need to limit the maximum transmit power on each antenna for it to operate within the linear region of the PA, which is more power efficient \cite{Zhang-J, YU}. The per-antenna power constraints (PAPCs) in \eqref{eq:PAPC}  different from those considered in \cite{Zhang:WSR:09,He:SRM:14}, but we note that the proposed solution introduced in this paper also applies to the sum power constraint (SPC) after slight modifications. This will be elaborated in the Appendix. In fact the interference from the secondary BS to a particular PU is a matrix, which is non-degrading.  This is different from the case of PUs with single antenna in which  the interference is a scalar, referred to as interference temperature \cite{Zhang:WSR:09}. The constraints in \eqref{eq:threshold} indicate that the sum of all eigenvalues of the resulting interference matrix should be less than a predetermined interference threshold $I_m$ at the $m$th PU.  We can rewrite \eqref{eq:threshold} as                    $\sum_{i=1}^{\tilde{n}_{m}}\sum_{k=1}^{K}\boldsymbol{g}_{m,i}\mathbf{T}_k\mathbf{T}_k^{H}\boldsymbol{g}_{m,i}^{H}\leq I_{m},$
where $\boldsymbol{g}_{m,i}\in\mathbb{C}^{1\times N} $ denotes the channel from the secondary BS to the $i$th receive antenna of the $m$th PU. We note that the term $ \sum_{k=1}^{K}\boldsymbol{g}_{m,i}\mathbf{T}_k\mathbf{T}_k^{H}\boldsymbol{g}_{m,i}^{H}$ represents the interference temperature at the $i$th antenna of the $m$th PU. In this way, \eqref{eq:threshold} implies that the sum of the interference temperatures of all antennas at the $m$th PU should be smaller than or equal to a predetermined threshold.
There are possibly other ways to control the interference generated by the secondary system. For example, we may force the maximum eigenvalue or the determinant (equivalent to the product of eigenvalues) of the interference matrix to be smaller than a threshold.  We remark that all these ways of controlling the interference term are the same for single-antenna PUs.

To further  simplify \eqref{eq:BD:Sumrate:PAPC},
let $\mathbf{\bar{H}}_{k}=[\mathbf{H}_{1}^{T}\cdots\mathbf{H}_{k-1}^{T}\,\mathbf{H}_{k+1}^{T}\cdots\mathbf{H}_{K}^{T}]^{T}\in\mathbb{C}^{(n_{R}-n_{k})\times N},$
where $n_{R}=\sum_{k=1}^{K}n_{k}$ is the total number of receive
antennas. The condition for \eqref{eq:zfcondi} to be feasible is that $\mathcal{N}(\bar{\mathbf{H}}_{k})$ has a dimension larger than zero, i.e., $n_{R}-n_{k}>0$, for all $k$. In this paper, we assume that the problem design in \eqref{eq:BD:Sumrate:PAPC} is feasible, which can be met when $N-\sum_{i\neq k}n_i\geq n_k$ and the channel matrices of the SUs and PUs have sufficiently low correlation degree. Let $\mathbf{\bar{V}}_{k}$ be the null space of $\bar{\mathbf{H}}_{k}$, then $\mathbf{\bar{V}}_{k}\in\mathbb{C}^{N\times\bar{n}_{k}}$ is a basis of $\mathcal{N}(\bar{\mathbf{H}}_{k})$, where $\bar{n}_{k}=N-\sum_{i\neq k}n_{i}$. Then we can write $\mathbf{T}_{k}=\bar{\mathbf{V}}_{k}\bar{\mathbf{T}}_{k}$, where $\bar{\mathbf{T}}_k$ are the solutions to the following problem
\begin{IEEEeqnarray}{rCl}\label{eq:BD:Sumrate:PAPC:ZFremoved}
\underset{\{\bar{\mathbf{T}}_{k}\}}{\maxi}\quad& & \sum_{k=1}^{K}\log|\mathbf{I}+\mathbf{H}_{k}\bar{\mathbf{V}}_{k}\bar{\mathbf{T}}_{k}\bar{\mathbf{T}}_{k}^{H}\bar{\mathbf{V}}_{k}^{H}\mathbf{H}_{k}^{H}| \IEEEyesnumber \IEEEyessubnumber\\
\st\quad& & \sum_{k=1}^{K}[\bar{\mathbf{V}}_{k}\bar{\mathbf{T}}_{k}\bar{\mathbf{T}}_{k}^{H}\bar{\mathbf{V}}_{k}^{H}]_{n,n}\leq P_{n},\:\forall n\in\mathcal{N}\IEEEyessubnumber\\
 & &\sum_{k=1}^{K}\tr(\tilde{\mathbf{G}}_{mk}\bar{\mathbf{T}}_{k}\bar{\mathbf{T}}_{k}^{H}\tilde{\mathbf{G}}_{mk}^{H})\leq I_{m},\forall m\in\mathcal{M}\IEEEyessubnumber
\end{IEEEeqnarray}
where $\tilde{\mathbf{G}}_{mk}\triangleq\mathbf{G}_{m}\bar{\mathbf{V}}_{k}\in\mathbb{C}^{\tilde{n}_{m}\times\bar{n}_{k}}$.
  Note that the problem \eqref{eq:BD:Sumrate:PAPC:ZFremoved} is not a convex program and indeed different from \cite{Tran:MISO:13},  where only conventional MISO systems were considered, i.e.,  without interference constraints. Moreover, we will solve \eqref{eq:BD:Sumrate:PAPC:relaxed} through a BC-MAC duality, in contrast to dealing with the primal domain as done in \cite{Tran:MIMO:13}. Towards this end we consider the following  problem.
\begin{IEEEeqnarray}{rCl}\label{eq:BD:Sumrate:PAPC:relaxed}
\underset{\{\mathbf{S}_{k}\}\succeq\mathbf{0}}{\maxi}\quad & &\sum_{k=1}^{K}\log|\mathbf{I}+\tilde{\mathbf{H}}_{k}\mathbf{S}_{k}\tilde{\mathbf{H}}_{k}^{H}|\IEEEyesnumber\IEEEyessubnumber\label{eq:eq:BD:Sumrate:PAPC:relaxed:obj}\\
\st\quad & &\sum_{k=1}^{K}[\bar{\mathbf{V}}_{k}\mathbf{S}_{k}\bar{\mathbf{V}}_{k}^{H}]_{n,n}\leq P_{n},\:\forall n\in\mathcal{N}\IEEEyessubnumber\label{eq:BD:Sumrate:PAPC:relaxed_a}\\
 & &\sum_{k=1}^{K}\tr(\tilde{\mathbf{G}}_{mk}\mathbf{S}_{k}\tilde{\mathbf{G}}_{mk}^{H})\leq I_{m},\forall m\in\mathcal{M}\IEEEyessubnumber\label{eq:BD:Sumrate:PAPC:relaxed_b}\label{eq:BD:Sumrate:Inter}\\
& &\rank(\mathbf{S}_{k})\leq n_{k},\,\forall k\IEEEyessubnumber\label{eq:constraint:rank}
\end{IEEEeqnarray}
where $\tilde{\mathbf{H}}_{k}=\mathbf{H}_{k}\mathbf{\bar{V}}_{k}\in\mathbb{C}^{n_{k}\times\bar{n}_{k}}$
and $\mathbf{S}_{k}=\bar{\mathbf{T}}_{k}\bar{\mathbf{T}}_{k}^{H}$. 
Though \eqref{eq:BD:Sumrate:PAPC:relaxed} is a nonconvex program, it can be solved to global optimality by dropping the rank constraints in \eqref{eq:constraint:rank} and then considering the so-called rank relaxed problem. We will prove shortly that the optimal solutions of the relaxed problem must satisfy the rank constraints in \eqref{eq:constraint:rank}, meaning that the rank relaxation is tight. Thus, from now onwards, we will consider the relaxed problem of \eqref{eq:BD:Sumrate:PAPC:relaxed} in the sequel of the paper, instead of \eqref{eq:BD:Sumrate:PAPC:relaxed} in the original form.

To solve the relaxed problem of \eqref{eq:BD:Sumrate:PAPC:relaxed}, we can modify the method presented in \cite{Zhang:BDCooperative:2010} which is based on the subgradient method. However, subgradient methods generally converge very slow in practice. Herein we propose an efficient method to solve it using a barrier method which is known to have superlinear convergence property.

\section{Proposed Algorithm}\label{Proposedalgorithm}
In this section, we first transform the problem \eqref{eq:BD:Sumrate:PAPC:relaxed} into an equivalent convex-concave problem \cite[Section 10.3.4]{Stephen}, thereby showing the structure of the optimal precoder $\mathbf{S}_k$. We then propose a numerical algorithm to solve the convex-concave problem by customizing interior-point methods to find a saddle point (i.e., an optimal solution).
The development of the proposed algorithm is  particularly based on the following theorem.
\begin{thm}
\label{thm:BC-MAC-duality}Consider the following convex-concave problem.
\begin{equation}
\begin{array}{rl}
\underset{\boldsymbol{\psi}\geq0}{\mini}\;\underset{\{\mathbf{Q}_{k}\}\succeq\mathbf{0}}{\maxi} & \sum_{k=1}^{K}\log\frac{|\bar{\mathbf{V}}_{k}^{H}\boldsymbol{\Lambda}\bar{\mathbf{V}}_{k}+\tilde{\mathbf{H}}_{k}^{H}\mathbf{Q}_{k}\tilde{\mathbf{H}}_{k}|}{|\bar{\mathbf{V}}_{k}^{H}\boldsymbol{\Lambda}\bar{\mathbf{V}}_{k}|}\\
\st & \sum_{k=1}^{K}\tr(\mathbf{Q}_{k})\leq P\\
 & \mathbf{p}^{T}\boldsymbol{\psi}\leq P.
\end{array}\label{eq:BD:Sumrate:dualMAC}
\end{equation}
 where $\mathbf{\Lambda}=\diag(\boldsymbol{\eta})+\sum_{m=1}^M\lambda_m\mathbf{G}_{m}^H\mathbf{G}_{m}$, $\boldsymbol{\psi}=[\boldsymbol{\eta}^T\quad\boldsymbol{\lambda}^T]^T$, and $\mathbf{p}=[\mathbf{\bar{p}}^T\quad \boldsymbol{\bar{I}}^T]^T$. We note that the objective function in \eqref{eq:BD:Sumrate:dualMAC} is convex in $\boldsymbol{\psi}$ for fixed $\{\mathbf{Q}_{k}\}$, and concave in $\{\mathbf{Q}_{k}\}$ for fixed $\boldsymbol{\psi}$, hence the name convex-concave. Let $\boldsymbol{\Omega}_{k}=\bar{\mathbf{V}}_{k}^{H}\boldsymbol{\Lambda}\bar{\mathbf{V}}_{k}$, and $\mathbf{U}_{k}\mathbf{D}_{k}\mathbf{V}_{k}^{H}$ be a singular value decomposition (SVD) of $\tilde{\mathbf{H}}_{k}\boldsymbol{\Omega}_{k}^{-1/2}$
where $\mathbf{D}_{k}$ is square and diagonal. Then, the optimal solution $\mathbf{S}_{k}$ of \eqref{eq:BD:Sumrate:PAPC:relaxed} can be obtained from that of \eqref{eq:BD:Sumrate:dualMAC} as
\begin{equation}
\mathbf{S}_{k}=\boldsymbol{\Omega}_{k}^{-1/2}\mathbf{V}_{k}\mathbf{U}_{k}^{H}\mathbf{Q}_{k}\mathbf{U}_{k}\mathbf{V}_{k}^{H}\boldsymbol{\Omega}_{k}^{-1/2}\label{eq:BD:duality}.
\end{equation}
\end{thm}
\begin{proof}
Please refer to the Appendix.
\end{proof}
\begin{rem}
Theorem 1 is different but similar in many ways to Theorem 2 in \cite{Tran:MISO:13} which is only dedicated to PAPCs and MISO cases.    Our duality result in Theorem \ref{thm:BC-MAC-duality} can be viewed as an extension of the duality results in  \cite{Tran:MISO:13} to the case of multiple linear constraints and MIMO cases.  Since $\rank(\mathbf{Q}_{k})\leq n_{k}$, it follows that $\rank(\mathbf{S}_{k})\leq n_{k}$, which means the relaxed problem of  \eqref{eq:BD:Sumrate:PAPC:relaxed} is equivalent to  \eqref{eq:BD:Sumrate:PAPC:ZFremoved}. 
\end{rem}

\begin{rem}
		Before proceeding further we provide some insights into  the convex-concave program in \eqref{eq:BD:Sumrate:dualMAC}. Let $\tilde{f}(\boldsymbol{\psi},\{\mathbf{Q}_{k}\})\triangleq \sum_{k=1}^{K}\log\frac{|\bar{\mathbf{V}}_{k}^{H}\boldsymbol{\Lambda}\bar{\mathbf{V}}_{k}+\tilde{\mathbf{H}}_{k}^{H}\mathbf{Q}_{k}\tilde{\mathbf{H}}_{k}|}{|\bar{\mathbf{V}}_{k}^{H}\boldsymbol{\Lambda}\bar{\mathbf{V}}_{k}|} $, i.e., the objective of \eqref{eq:BD:Sumrate:dualMAC}. Then a problem  closely related to \eqref{eq:BD:Sumrate:dualMAC} is  given by 
	\begin{equation}
	\begin{array}{rl}
	\underset{\{\mathbf{Q}_{k}\}\succeq\mathbf{0}}{\maxi}\;\underset{\boldsymbol{\psi}\geq0}{\mini}& \tilde{f}(\boldsymbol{\psi},\{\mathbf{Q}_{k}\})\\
	\st & \sum_{k=1}^{K}\tr(\mathbf{Q}_{k})\leq P\\
	& \mathbf{p}^{T}\boldsymbol{\psi}\leq P.
	\end{array}\label{eq:BD:Sumrate:maximin}
	\end{equation}	
We can say that $(\boldsymbol{\psi}^{\ast},\{\mathbf{Q}_{k}^{\ast}\})$ is a solution to the convex-concave program or a saddle-point for the problem, if  for all $\boldsymbol{\psi}$ and $\{\mathbf{Q}_{k}\}$ 
\begin{equation}
\tilde{f}(\boldsymbol{\psi}^{\ast},\{\mathbf{Q}_{k}\})\leq \tilde{f}(\boldsymbol{\psi}^{\ast},\{\mathbf{Q}_{k}^{\ast}\})\leq\tilde{f}(\boldsymbol{\psi},\{\mathbf{Q}_{k}^{\ast}\}).
\end{equation}
Since $\tilde{f}(\boldsymbol{\psi},\{\mathbf{Q}_{k}\})$ is differentiable, the above inequality implies that the strong max-min property holds, i.e., 
\begin{equation}
\underset{\{\mathbf{Q}_{k}\}}{\maxi}\;\underset{\boldsymbol{\psi}}{\mini}\; \tilde{f}(\boldsymbol{\psi},\{\mathbf{Q}_{k}\}) = \underset{\boldsymbol{\psi}}{\mini}\;\underset{\{\mathbf{Q}_{k}\}}{\maxi} \;\tilde{f}(\boldsymbol{\psi},\{\mathbf{Q}_{k}\}).
\end{equation}
\end{rem}

It is clear from the above discussions that solving \eqref{eq:BD:Sumrate:dualMAC} boils down to finding a saddle point for the convex-concave problem  for which we will propose a computationally efficient algorithm. First, as intermediate results when proving Theorem 1 (see the steps from \eqref{eq:appendix:38} to \eqref{eq:BD:Sumrate:dualMAC:proof:dual} in the Appendix), we can set the inequality constraints to be equality ones without affecting the optimality. Based on this observation, it is more computationally efficient\footnote{The equality constraints are generally easier to handle when using a barrier method to solve an optimization problem. The reason is that we need to introduce a barrier function (i.e. the log function in our case) to deal with inequality constrains, while we do not need to do so for equality constraints.} to consider the following problem rather than  problem \eqref{eq:BD:Sumrate:dualMAC}.
\begin{equation}
\begin{array}{rl}
\underset{\boldsymbol{\psi}\geq0}{\mini}\;\underset{\{\mathbf{Q}_{k}\}\succeq\mathbf{0}}{\maxi} & \sum_{k=1}^{K}\log\frac{|\bar{\mathbf{V}}_{k}^{H}\boldsymbol{\Lambda}\bar{\mathbf{V}}_{k}+\tilde{\mathbf{H}}_{k}^{H}\mathbf{Q}_{k}\tilde{\mathbf{H}}_{k}|}{|\bar{\mathbf{V}}_{k}^{H}\boldsymbol{\Lambda}\bar{\mathbf{V}}_{k}|}\\
\st & \sum_{k=1}^{K}\tr(\mathbf{Q}_{k})=P\\
 & \mathbf{p}^{T}\boldsymbol{\psi}=P.
\end{array}\label{eq:BD:Sumrate:dualMAC:equi}
\end{equation}
 The proposed method is a result of applying a barrier method to find a saddle point for problem  \eqref{eq:BD:Sumrate:dualMAC:equi}. According to the barrier method, we consider the modified objective function given by
\begin{IEEEeqnarray}{rCl}\begin{IEEEeqnarraybox}[][c]{rCl}f(t,\boldsymbol{\psi},\{\mathbf{Q}_{k}\}) & = & \sum_{k=1}^{K}\log\frac{|\bar{\mathbf{V}}_{k}^{H}\boldsymbol{\Lambda}\bar{\mathbf{V}}_{k}+\tilde{\mathbf{H}}_{k}^{H}\mathbf{Q}_{k}\tilde{\mathbf{H}}_{k}|}{|\bar{\mathbf{V}}_{k}^{H}\boldsymbol{\Lambda}\bar{\mathbf{V}}_{k}|} \\  &  & \;-\frac{1}{t}\sum_{i=1}^{N+M}\log(\psi_{i})+\frac{1}{t}\sum_{k=1}^{K}\log|\mathbf{Q}_{k}|
\end{IEEEeqnarraybox}\label{eq:BD:Sumrate:dualMAC:modi}\IEEEeqnarraynumspace\end{IEEEeqnarray}
where $\log|\mathbf{Q}_{k}|$ and  $\log(\psi_{i})$ are the logarithmic barrier functions accounting for the positive semidefinite matrix constraint $\mathbf{Q}_k\succeq\mathbf{0}$ and  $\psi_{i}\geq 0$, respectively, and $t > 0$ is a parameter that controls the logarithmic barrier terms. 
Given a fixed value $t$, the barrier method requires solving the following equality constrained minimax problem
\begin{subequations}\begin{align}
\underset{\boldsymbol{\psi}\geq0}{\mini}\;\underset{\{\mathbf{Q}_{k}\}\succeq\mathbf{0}}{\maxi} & f(t,\boldsymbol{\psi},\{\mathbf{Q}_{k}\})\\
\st & \sum_{k=1}^{K}\tr(\mathbf{Q}_{k})=P\label{eq:BD:Sumrate:dualMAC:equa:b}\\
 & \mathbf{p}^{T}\boldsymbol{\psi}=P.\label{eq:BD:Sumrate:dualMAC:equa:c}
\end{align}\label{eq:BD:Sumrate:dualMAC:equa:t}\end{subequations}
 $\indent$The main idea of interior-point method is that, given a fixed value $t$, we find an optimal solution $(\{\mathbf{Q}_k\}, \boldsymbol{\psi})$ to \eqref{eq:BD:Sumrate:dualMAC:equa:t} which is referred to as the centering step, and then increase $t$ until the duality gap of the minimax optimization problem in \eqref{eq:BD:Sumrate:dualMAC:equa:t} satisfies an accuracy level.  We start with KKT conditions for \eqref{eq:BD:Sumrate:dualMAC:equa:t}, which are shown at the top of the next page, where  $\mathbf{v}_{k,i}=\bar{\mathbf{V}}_{k}^{H}\mathbf{e}_{i}$, i.e., $\mathbf{v}_{k,i}$ is the $i$th column of $\bar{\mathbf{V}}_{k}^{H}$, and $\mu_i, \; i=\{1, 2\}$, are the dual variables corresponding to the constraints in \eqref{eq:BD:Sumrate:dualMAC:equa:b} and \eqref{eq:BD:Sumrate:dualMAC:equa:c}, respectively.
 In \eqref{eq:stationary:Qk}, we have utilized the fact that the gradient of $\sum_{k=1}^{K}\log\frac{|\bar{\mathbf{V}}_{k}^{H}\boldsymbol{\Lambda}\bar{\mathbf{V}}_{k}+\tilde{\mathbf{H}}_{k}^{H}\mathbf{Q}_{k}\tilde{\mathbf{H}}_{k}|}{|\bar{\mathbf{V}}_{k}^{H}\boldsymbol{\Lambda}\bar{\mathbf{V}}_{k}|}$ with respect to $\mathbf{Q}_k$ is given by $\nabla_{\mathbf{Q}_k}\sum_{k=1}^{K}\log\frac{|\bar{\mathbf{V}}_{k}^{H}\boldsymbol{\Lambda}\bar{\mathbf{V}}_{k}+\tilde{\mathbf{H}}_{k}^{H}\mathbf{Q}_{k}\tilde{\mathbf{H}}_{k}|}{|\bar{\mathbf{V}}_{k}^{H}\boldsymbol{\Lambda}\bar{\mathbf{V}}_{k}|}=\tilde{\mathbf{H}}_{k}\bigl(\bar{\mathbf{V}}_{k}^{H}\boldsymbol{\Lambda}\bar{\mathbf{V}}_{k}+\tilde{\mathbf{H}}_{k}^{H}\mathbf{Q}_{k}\tilde{\mathbf{H}}_{k}\bigr)^{-1}\tilde{\mathbf{H}}_{k}^{H}$. Similarly, the gradient of $\sum_{k=1}^{K}\log\frac{|\bar{\mathbf{V}}_{k}^{H}\boldsymbol{\Lambda}\bar{\mathbf{V}}_{k}+\tilde{\mathbf{H}}_{k}^{H}\mathbf{Q}_{k}\tilde{\mathbf{H}}_{k}|}{|\bar{\mathbf{V}}_{k}^{H}\boldsymbol{\Lambda}\bar{\mathbf{V}}_{k}|}$ with respect to $\boldsymbol{\psi} (\boldsymbol{\eta}, \boldsymbol{\lambda})$, where $\boldsymbol{\Lambda}$ is a function of $\boldsymbol{\eta}$ and $ \boldsymbol{\lambda}$, is obtained as the first term in \eqref{eq:stationary:eta} and \eqref{eq:stationary:lambda}.

\begin{figure*}[t]\flushright{\begin{IEEEeqnarray}{rCl}\label{eq:KKT}
\tilde{\mathbf{H}}_{k}\bigl(\bar{\mathbf{V}}_{k}^{H}\boldsymbol{\Lambda}\bar{\mathbf{V}}_{k}+\tilde{\mathbf{H}}_{k}^{H}\mathbf{Q}_{k}\tilde{\mathbf{H}}_{k}\bigr)^{-1}\tilde{\mathbf{H}}_{k}^{H}+\frac{1}{t}\mathbf{Q}_{k}^{-1}-\mu_{1}\mathbf{I}&=&\mathbf{0},\forall k \IEEEyesnumber\IEEEyessubnumber \label{eq:stationary:Qk}\\
\sum_{k=1}^{K}\tr(\mathbf{Q}_{k})&=&P\IEEEyessubnumber \label{eq:feasible:Qk}\\
\sum_{k=1}^{K}\mathbf{v}_{k,i}^{H}\Bigl[(\bar{\mathbf{V}}_{k}^{H}\boldsymbol{\Lambda}\bar{\mathbf{V}}_{k}+\tilde{\mathbf{H}}_{k}^{H}\mathbf{Q}_{k}\tilde{\mathbf{H}}_{k})^{-1}-(\bar{\mathbf{V}}_{k}^{H}\boldsymbol{\Lambda}\bar{\mathbf{V}}_{k})^{-1}\Bigr]\mathbf{v}_{k,i}-\frac{1}{t}\eta_{i}^{-1}+\mu_{2}P_{i}&=&0,\forall i\IEEEyessubnumber\label{eq:stationary:eta}\\
\sum_{k=1}^{K}\tr\bigl(\tilde{\mathbf{G}}_{mk}\bigl[(\bar{\mathbf{V}}_{k}^{H}\boldsymbol{\Lambda}\bar{\mathbf{V}}_{k}+\tilde{\mathbf{H}}_{k}^{H}\mathbf{Q}_{k}\tilde{\mathbf{H}}_{k})^{-1}-(\bar{\mathbf{V}}_{k}^{H}\boldsymbol{\Lambda}\bar{\mathbf{V}}_{k})^{-1}\bigr]\tilde{\mathbf{G}}_{mk}^{H}\bigr)-\frac{\lambda_{m}^{-1}}{t}+\mu_{2}I_{m}&=&0,\forall m \IEEEeqnarraynumspace \IEEEyessubnumber\label{eq:stationary:lambda}\\
\mathbf{p}^{T}\boldsymbol{\psi}&=&P.\IEEEyessubnumber \label{eq:feasible:psi} 
\end{IEEEeqnarray}}
\hrulefill{}
\end{figure*}

To find a solution to the system of KKT conditions in \eqref{eq:KKT}, we use the infeasible start Newton method. More explicitly, the proposed algorithm starts with a point that does not satisfy the equalities. The key computation is to find a Newton step  in each iteration method. Towards this end we replace  $\mathbf{Q}_{k}$ by $\mathbf{Q}_{k}+\Delta\mathbf{Q}_{k}$, $\boldsymbol{\psi}$ by $\boldsymbol{\psi}+\Delta\boldsymbol{\psi}$, and $\mu_{i}$ by $\mu_{i}+\Delta\mu_{i}$
in \eqref{eq:stationary:Qk} to obtain
\begin{equation}\small\begin{aligned}&t\tilde{\mathbf{H}}_{k}\Bigl(\bar{\mathbf{V}}_{k}^{H}\boldsymbol{\Lambda}\bar{\mathbf{V}}_{k}+\tilde{\mathbf{H}}_{k}^{H}\mathbf{Q}_{k}\tilde{\mathbf{H}}_{k}+\bar{\mathbf{V}}_{k}^{H}\Delta\boldsymbol{\Lambda}\bar{\mathbf{V}}_{k}+ 
\tilde{\mathbf{H}}_{k}^{H}\Delta\mathbf{Q}_{k}\tilde{\mathbf{H}}_{k}\Bigr)^{-1}\tilde{\mathbf{H}}_{k}^{H}\\
& +(\mathbf{Q}_{k}+\Delta\mathbf{Q}_{k})^{-1}\negmedspace\negmedspace-t(\mu_{1}+\Delta\mu_{1})\mathbf{I}  =  \mathbf{0}\IEEEeqnarraynumspace\label{eq:statationary:Qk:step}\end{aligned}\end{equation}
for all $k$, where $\Delta\boldsymbol{\Lambda}=\diag(\boldsymbol{\Delta\eta})+\sum_{m=1}^M\Delta\lambda_m\mathbf{G}_{m}^H\mathbf{G}_{m}$ and $\Delta\boldsymbol{\psi}=[\Delta\boldsymbol{\eta}^T\quad\Delta\boldsymbol{\lambda}^T]^T$. By using the identity $(\boldsymbol{A}+\boldsymbol{B})^{-1}\approx\boldsymbol{A}^{-1}-\boldsymbol{A}^{-1}\boldsymbol{B}\boldsymbol{A}^{-1}$
for small $\boldsymbol{B}$,\footnote{The approximation is  precise for small $\boldsymbol{B}$ and relatively crude for large $\boldsymbol{B}$. In fact, in the first  iterations, the Newton steps  are large and the approximation is not very accurate.  However, when the algorithm approaches the optimal solution, the Newton steps will become  small and thus the approximation is very accurate. This will lead to a superlinear convergence rate of the proposed algorithm as demonstrated in Fig.~\ref{fig:Convergencebehavior:Iteration:N}.} and defining $\mathbf{\dot{H}}_{k}=\tilde{\mathbf{H}}_{k}\Bigl(\bar{\mathbf{V}}_{k}^{H}\boldsymbol{\Lambda}\bar{\mathbf{V}}_{k}+\tilde{\mathbf{H}}_{k}^{H}\mathbf{Q}_{k}\tilde{\mathbf{H}}_{k}\Bigr)^{-1}\tilde{\mathbf{H}}_{k}^{H}$ and
$\hat{\mathbf{H}}_{k}=\tilde{\mathbf{H}}_{k}\Bigl(\bar{\mathbf{V}}_{k}^{H}\boldsymbol{\Lambda}\bar{\mathbf{V}}_{k}+\tilde{\mathbf{H}}_{k}^{H}\mathbf{Q}_{k}\tilde{\mathbf{H}}_{k}\Bigr)^{-1}\bar{\mathbf{V}}_{k}^{H}$, we can approximate \eqref{eq:statationary:Qk:step} as
\begin{equation}\begin{aligned}
&t\mathbf{Q}_{k}\hat{\mathbf{H}}_{k}\Delta\boldsymbol{\Lambda}\hat{\mathbf{H}}_{k}^{H}\mathbf{Q}_{k}+t\mathbf{Q}_{k}\mathbf{\dot{H}}_{k}\Delta\mathbf{Q}_{k}\mathbf{\dot{H}}_{k}\mathbf{Q}_{k}\\
&+\Delta\mathbf{Q}_{k}+t\Delta\mu_{1}\mathbf{Q}_{k}^{2}=t\mathbf{Q}_{k}\mathbf{\dot{H}}_{k}\mathbf{Q}_{k}-t\mu_{1}\mathbf{Q}_{k}^{2}+\mathbf{Q}_{k}.\label{eq:stationary:Qk:approx}
\end{aligned}\end{equation}
Note that
\begin{equation}\begin{aligned}
\mathbf{Q}_{k}\hat{\mathbf{H}}_{k}\Delta\boldsymbol{\Lambda}\hat{\mathbf{H}}_{k}^{H}\mathbf{Q}_{k}&=\sum_{i=1}^{N}\Delta\eta_{i}\mathbf{Q}_{k}\hat{\mathbf{H}}_{k}\mathbf{e}_{i}\mathbf{e}_{i}^{H}\hat{\mathbf{H}}_{k}^{H}\mathbf{Q}_{k}\\
&\quad+\mathbf{Q}_{k}\hat{\mathbf{H}}_{k}\sum_{m=1}^{M}\Delta\lambda_m\mathbf{G}_{m}^H\mathbf{G}_{m}\hat{\mathbf{H}}_{k}^{H}\mathbf{Q}_{k}\\
&=\sum_{i=1}^{N}\Delta\eta_{i}\mathbf{u}_{k,i}\mathbf{u}_{k,i}^{H}\\
&\quad+\mathbf{Q}_{k}\hat{\mathbf{H}}_{k}\sum_{m=1}^{M}\Delta\lambda_m\mathbf{G}_{m}^H\mathbf{G}_{m}\hat{\mathbf{H}}_{k}^{H}\mathbf{Q}_{k}\nonumber
\end{aligned}\end{equation}
where $\mathbf{u}_{k,i}=\mathbf{Q}_{k}\hat{\mathbf{H}}_{k}\mathbf{e}_{i}$. Thus, \eqref{eq:stationary:Qk:approx} is equal to
\begin{equation}\begin{aligned}
&t\sum_{i=1}^{N}\Delta\eta_{i}\mathbf{u}_{k,i}\mathbf{u}_{k,i}^{H}+t\mathbf{Q}_{k}\hat{\mathbf{H}}_{k}\sum_{m=1}^{M}\Delta\lambda_m\mathbf{G}_{m}^H\mathbf{G}_{m}\hat{\mathbf{H}}_{k}^{H}\mathbf{Q}_{k}\\
&+t\mathbf{Q}_{k}\mathbf{\dot{H}}_{k}\Delta\mathbf{Q}_{k}\mathbf{\dot{H}}_{k}\mathbf{Q}_{k}+\Delta\mathbf{Q}_{k}+t\Delta\mu_{1}\mathbf{Q}_{k}^{2}\\
&=t\mathbf{Q}_{k}\mathbf{\dot{H}}_{k}\mathbf{Q}_{k}-t\mu_{1}\mathbf{Q}_{k}^{2}+\mathbf{Q}_{k}.\label{eq:stationary:Qk:approx:compress}
\end{aligned}\end{equation}
Next, from \eqref{eq:feasible:Qk} and \eqref{eq:feasible:psi},
we have
\begin{eqnarray}
\sum_{k=1}^{K}\tr(\Delta\mathbf{Q}_{k}) & = & P-\sum_{k=1}^{K}\tr(\mathbf{Q}_{k})\label{eq:feasible:Qk:step}\\
\mathbf{p}^{T}\Delta\boldsymbol{\psi} & = & P-\mathbf{p}^{T}\boldsymbol{\psi}.\label{eq:feasible:psi:step}
\end{eqnarray}
From \eqref{eq:stationary:eta} and \eqref{eq:stationary:lambda}, we have \eqref{eq:stationary:eta:gradient} and \eqref{eq:stationary:lambda:gradient}  shown at the top of the next page.
\begin{figure*}
\centering{\begin{IEEEeqnarray}{rCl}
t\sum_{k=1}^{K}\mathbf{v}_{k,i}^{H}\Bigl[(\bar{\mathbf{V}}_{k}^{H}\boldsymbol{\Lambda}\bar{\mathbf{V}}_{k}+\tilde{\mathbf{H}}_{k}^{H}\mathbf{Q}_{k}\tilde{\mathbf{H}}_{k}+\bar{\mathbf{V}}_{k}^{H}\Delta\boldsymbol{\Lambda}\bar{\mathbf{V}}_{k}+\tilde{\mathbf{H}}_{k}^{H}\Delta\mathbf{Q}_{k}\tilde{\mathbf{H}}_{k})^{-1}& &\label{eq:stationary:eta:gradient}\\
-(\bar{\mathbf{V}}_{k}^{H}\boldsymbol{\Lambda}\bar{\mathbf{V}}_{k}+\bar{\mathbf{V}}_{k}^{H}\Delta\boldsymbol{\Lambda}\bar{\mathbf{V}}_{k})^{-1}\Bigr]\mathbf{v}_{k,i}-(\eta_{i}+\Delta\eta_{i})^{-1}+t(\mu_{2}+\Delta\mu_{2})P_{i}&=& 0,\;\forall i.\nonumber\\
t\sum_{k=1}^{K}\tr\Bigl(\mathbf{\tilde{G}}_{mk}\Bigl[(\bar{\mathbf{V}}_{k}^{H}\boldsymbol{\Lambda}\bar{\mathbf{V}}_{k}+\tilde{\mathbf{H}}_{k}^{H}\mathbf{Q}_{k}\tilde{\mathbf{H}}_{k}+\bar{\mathbf{V}}_{k}^{H}\Delta\boldsymbol{\Lambda}\bar{\mathbf{V}}_{k}+\tilde{\mathbf{H}}_{k}^{H}\Delta\mathbf{Q}_{k}\tilde{\mathbf{H}}_{k})^{-1}\label{eq:stationary:lambda:gradient}\\
-(\bar{\mathbf{V}}_{k}^{H}\boldsymbol{\Lambda}\bar{\mathbf{V}}_{k}+\bar{\mathbf{V}}_{k}^{H}\Delta\boldsymbol{\Lambda}\bar{\mathbf{V}}_{k})^{-1}\Bigr]\mathbf{\tilde{G}}_{mk}^{H}\Bigr)
-(\lambda_{m}+\Delta\lambda_{m})^{-1}+t(\mu_{2}+\Delta\mu_{2})I_{m}  &=&  0,\;\forall m.\nonumber
\end{IEEEeqnarray}}\hrulefill{}
\end{figure*}

Letting $\boldsymbol{\Pi}_k=	\bigl(\bar{\mathbf{V}}_{k}^{H}\boldsymbol{\Lambda}\bar{\mathbf{V}}_{k}+\tilde{\mathbf{H}}_{k}^{H}\mathbf{Q}_{k}\tilde{\mathbf{H}}_{k}\bigr)$  and following the same steps from \eqref{eq:statationary:Qk:step} to \eqref{eq:stationary:Qk:approx:compress} we can approximate \eqref{eq:stationary:eta:gradient} as
\begin{IEEEeqnarray}{rCl}
&&t(\varphi_{i}-\sum_{j=1}^{N}\varphi_{i,j}\Delta\eta_{j}-\sum_{m=1}^{M}\varphi_{i,m}\Delta\lambda_{m}-\sum_{k=1}^{K}\boldsymbol{\beta}_{k,i}^{H}\Delta\mathbf{Q}_{k}\boldsymbol{\beta}_{k,i})\nonumber\\
&&+\eta_{i}^{-2}\Delta\eta_{i}+t\Delta\mu_{2}P_i=\eta_{i}^{-1}-t\mu_{2}P_i,\;\forall i\label{eq:stationary:psi:step:compress1}
\end{IEEEeqnarray}
where $\varphi_{i}$, $\varphi_{i,j}$, $\varphi_{i,m}$, and $\boldsymbol{\beta}_{k,i}$ are, respectively, given by
\begin{equation}\begin{aligned}
\varphi_{i}&=\sum_{k=1}^{K}\mathbf{v}_{k,i}^{H}\bigl[\boldsymbol{\Pi}_k^{-1}-\bigl(\bar{\mathbf{V}}_{k}^{H}\boldsymbol{\Lambda}\bar{\mathbf{V}}_{k}\bigr)^{-1}\bigr]\mathbf{v}_{k,i}\\
\varphi_{i,j}&=\sum_{k=1}^{K}|\mathbf{v}_{k,i}^{H}\boldsymbol{\Pi}_k^{-1}\bar{\mathbf{V}}_{k}^{H}\mathbf{e}_{j}|^{2}-|\mathbf{v}_{k,i}^{H}\bigl(\bar{\mathbf{V}}_{k}^{H}\boldsymbol{\Lambda}\bar{\mathbf{V}}_{k}\bigr)^{-1}\bar{\mathbf{V}}_{k}^{H}\mathbf{e}_{j}|^{2}\\
\varphi_{i,m}&=\sum_{k=1}^{K}\|\mathbf{v}_{k,i}^{H}\boldsymbol{\Pi}_k^{-1}\tilde{\mathbf{G}}_{mk}^H\|^{2}-\|\mathbf{v}_{k,i}^{H}\bigl(\bar{\mathbf{V}}_{k}^{H}\boldsymbol{\Lambda}\bar{\mathbf{V}}_{k}\bigr)^{-1}\tilde{\mathbf{G}}_{mk}^H\|^{2}\\
\boldsymbol{\beta}_{k,i}&=\tilde{\mathbf{H}}_{k}\boldsymbol{\Pi}_k^{-1}\mathbf{v}_{k,i}.
\end{aligned}\end{equation}
In the same way,  \eqref{eq:stationary:lambda:gradient} is approximated as
\begin{equation}\begin{aligned}
&t\Bigl(\phi_{m}-\negmedspace\negmedspace \sum_{j=1}^{N}\negmedspace\phi_{m,j}\Delta\eta_{j}-\negmedspace\negmedspace\sum_{s=1}^{M}\negmedspace\phi_{m,s}\Delta\lambda_{s}-\negmedspace\negmedspace\sum_{k=1}^{K}\negmedspace\tr\bigl(\boldsymbol{\Xi}_{m,k}^{H}\Delta\mathbf{Q}_{k}\boldsymbol{\Xi}_{m,k}\bigr)\Bigr)\\
&+\lambda_{m}^{-2}\Delta\lambda_{m}+t\Delta\mu_{2}I_m=\lambda_{m}^{-1}-t\mu_{2}I_m,\;\forall m
\end{aligned}\label{eq:stationary:psi:step:compress1_a}
\end{equation}
where $\phi_{m}$, $\phi_{m,j}$, $\phi_{m,s}$, and $\boldsymbol{\Xi}_{m,k}$ are, respectively, defined as
\begin{equation}\begin{aligned}
\phi_{m}&=\sum_{k=1}^{K}\tr\Bigl(\tilde{\mathbf{G}}_{mk}\bigl[\boldsymbol{\Pi}_k^{-1}-\bigl(\bar{\mathbf{V}}_{k}^{H}\boldsymbol{\Lambda}\bar{\mathbf{V}}_{k}\bigr)^{-1}\bigr]\tilde{\mathbf{G}}_{mk}^H\Bigr)\\
\phi_{m,j}&=\sum_{k=1}^{K}\tr\Bigl(\tilde{\mathbf{G}}_{mk}\Bigl[\boldsymbol{\Pi}_k^{-1}\bar{\mathbf{V}}_{k}^{H}\mathbf{e}_{j}\mathbf{e}_{j}^H\bar{\mathbf{V}}_{k}\boldsymbol{\Pi}_k^{-1}\\
&\quad-\bigl(\bar{\mathbf{V}}_{k}^{H}\boldsymbol{\Lambda}\bar{\mathbf{V}}_{k}\bigr)^{-1}\bar{\mathbf{V}}_{k}^{H}\mathbf{e}_{j}\mathbf{e}_{j}^H\mathbf{V}_{k}\bigl(\bar{\mathbf{V}}_{k}^{H}\boldsymbol{\Lambda}\bar{\mathbf{V}}_{k}\bigr)^{-1}\Bigr]\tilde{\mathbf{G}}_{mk}^H\Bigr)\\
\phi_{m,s}&=\sum_{k=1}^{K}\tr\Bigl(\tilde{\mathbf{G}}_{mk}\Bigl[\boldsymbol{\Pi}_k^{-1}\tilde{\mathbf{G}}_{sk}^H\tilde{\mathbf{G}}_{sk}\boldsymbol{\Pi}_k^{-1}\\
&\quad-\linebreak\bigl(\bar{\mathbf{V}}_{k}^{H}\boldsymbol{\Lambda}\bar{\mathbf{V}}_{k}\bigr)^{-1}\tilde{\mathbf{G}}_{sk}^H\tilde{\mathbf{G}}_{sk}\bigl(\bar{\mathbf{V}}_{k}^{H}\boldsymbol{\Lambda}\bar{\mathbf{V}}_{k}\bigr)^{-1}\Bigr]\tilde{\mathbf{G}}_{mk}^H\Bigr)\\
\boldsymbol{\Xi}_{m,k}&=\tilde{\mathbf{H}}_{k}\boldsymbol{\Pi}_k^{-1}\tilde{\mathbf{G}}_{mk}^H.
\end{aligned}\end{equation}

We can find the Newton steps for the optimization variables by staking \eqref{eq:stationary:Qk:approx:compress}, \eqref{eq:feasible:Qk:step}, \eqref{eq:feasible:psi:step}, \eqref{eq:stationary:psi:step:compress1}, and \eqref{eq:stationary:psi:step:compress1_a} into a system of linear equations.  However such a conventional method requires complexity of $\mathcal{O}\bigl(K^3N^6\bigr)$\footnote{For a complex Hermitian matrix, the number of variables is $N^2$ where $N$ is the size of the matrix. Suppose we have $K$ Hermitian matrices, then the total number of variables is $K\times N^2$. Thus solving the system of linear equations has  complexity of $\mathcal{O}(K^3N^6)$.} which is relatively high. In this paper we instead follow a block elimination method \cite{Tran:MIMO:13,Stephen} which results in much lower complexity.
\begin{equation}
\Delta\mathbf{Q}_{k}=\boldsymbol{\Sigma}_{k}^{(0)}+\sum_{i=1}^{N+M}\Delta\psi_{i}\boldsymbol{\Sigma}_{k}^{(i)}+\Delta\mu_{1}\boldsymbol{\Sigma}_{k}^{(N+M+1)}.\label{eq:block:elimination}
\end{equation}
For notational simplicity, in \eqref{eq:discrete:Sylvester}, \eqref{eq:stationary:psi:step:compress2a}, and \eqref{eq:stationary:psi:step:compress2b}, we explicitly write \eqref{eq:block:elimination} as $\Delta\mathbf{Q}_{k}=\boldsymbol{\Sigma}_{k}^{(0)}+\sum_{i=1}^{N}\Delta\eta_{i}\boldsymbol{\Sigma}_{k}^{(i)}+\sum_{m=1}^{M}\Delta\lambda_{m}\boldsymbol{\Sigma}_{k}^{(m)}+\Delta\mu_{1}\boldsymbol{\Sigma}_{k}^{(N+M+1)}$. Substituting \eqref{eq:block:elimination} into \eqref{eq:stationary:Qk:approx:compress} yields a system of $(N+M+2)$ discrete-time Sylvester equations as follows:
\begin{equation}\begin{aligned}
&t\mathbf{Q}_{k}\mathbf{\dot{H}}_{k}\boldsymbol{\Sigma}_{k}^{(0)}\mathbf{\dot{H}}_{k}\mathbf{Q}_{k}+\boldsymbol{\Sigma}_{k}^{(0)}=t\mathbf{Q}_{k}\mathbf{\dot{H}}_{k}\mathbf{Q}_{k}+\mathbf{Q}_{k}-t\mu_{1}\mathbf{Q}_{k}^{2}\\
&t\mathbf{Q}_{k}\mathbf{\dot{H}}_{k}\boldsymbol{\Sigma}_{k}^{(i)}\mathbf{\dot{H}}_{k}\mathbf{Q}_{k}+\boldsymbol{\Sigma}_{k}^{(i)}=-t\mathbf{u}_{k,i}\mathbf{u}_{k,i}^{H},\;\text{for}\; i=1,\cdots, N\\
&t\mathbf{Q}_{k}\mathbf{\dot{H}}_{k}\boldsymbol{\Sigma}_{k}^{(m)}\mathbf{\dot{H}}_{k}\mathbf{Q}_{k}+\boldsymbol{\Sigma}_{k}^{(m)}=-t\mathbf{Q}_{k}\hat{\mathbf{H}}_{k}\mathbf{G}_{m}^H\mathbf{G}_{m}\hat{\mathbf{H}}_{k}^{H}\mathbf{Q}_{k}, \forall m\\
&t\mathbf{Q}_{k}\mathbf{\dot{H}}_{k}\boldsymbol{\Sigma}_{k}^{(N+M+1)}\mathbf{\dot{H}}_{k}\mathbf{Q}_{k}+\boldsymbol{\Sigma}_{k}^{(N+M+1)}  =  -t\mathbf{Q}_{k}^{2}.\label{eq:discrete:Sylvester}
\end{aligned}\end{equation}
Numerical methods to solve the discrete-time Sylvester equations in \eqref{eq:discrete:Sylvester} with complexity $\mathcal{O}(n_{k}^{3})$ can be found, e.g., in \cite{Higham}. That is to say, the complexity of solving \eqref{eq:discrete:Sylvester} is much less than that of solving a system of linear equations. Then, substituting \eqref{eq:block:elimination} into \eqref{eq:stationary:psi:step:compress1}
results in
\begin{IEEEeqnarray}{rCl}
&&t\Bigl(\varphi_{i}-\sum_{j=1}^{N}\varphi_{i,j}\Delta\eta_{j}-\sum_{m=1}^{M}\varphi_{i,m}\Delta\lambda_{m}-\sum_{k=1}^{K}\boldsymbol{\beta}_{k,i}^{H}\bigl(\boldsymbol{\Sigma}_{k}^{(0)}\nonumber\\
&&+\sum_{j=1}^{N}\Delta\eta_{j}\boldsymbol{\Sigma}_{k}^{(j)}+\sum_{m=1}^{M}\Delta\lambda_{m}\boldsymbol{\Sigma}_{k}^{(m)}
+\Delta\mu_{1}\boldsymbol{\Sigma}_{k}^{(N+M+1)}\bigr)\boldsymbol{\beta}_{k,i}\Bigr)\nonumber\\
&&+\eta_{i}^{-2}\Delta\eta_{i}+tP_{i}\Delta\mu_{2}=\eta_{i}^{-1}-tP_{i}\mu_{2}.\label{eq:stationary:psi:step:compress2a}
\end{IEEEeqnarray}
Let $\omega_{i}=\sum_{k=1}^{K}\boldsymbol{\beta}_{k,i}^{H}\boldsymbol{\Sigma}_{k}^{(N+M+1)}\boldsymbol{\beta}_{k,i}$,
$\gamma_{i,j}=\sum_{k=1}^{K}\boldsymbol{\beta}_{k,i}^{H}\boldsymbol{\Sigma}_{k}^{(j)}\boldsymbol{\beta}_{k,i}$, and $\gamma_{i,m}=\sum_{k=1}^{K}\boldsymbol{\beta}_{k,i}^{H}\boldsymbol{\Sigma}_{k}^{(m)}\boldsymbol{\beta}_{k,i}$.
Then, we can rewrite \eqref{eq:stationary:psi:step:compress2a} as
\begin{IEEEeqnarray}{rCl}
&t\sum_{j=1}^{N}\tilde{\varphi}_{i,j}\Delta\eta_{j}+t\sum_{m=1}^{M}\tilde{\varphi}_{i,m}\Delta\lambda_{m}
-\eta_{i}^{-2}\Delta\eta_{i}+t\omega_{i}\Delta\mu_{1}\nonumber\\
&-tP_{i}\Delta\mu_{2}=t\tilde{\varphi}_{i}+tP_{i}\mu_{2}-\eta_{i}^{-1},\,\forall i=1,\,2,\ldots,\, N\label{eq:stationary:psi:step:compress3}
\end{IEEEeqnarray}
where $\tilde{\varphi}_{i,j}=\varphi_{i,j}+\gamma_{i,j}$, $\tilde{\varphi}_{i,m}=\varphi_{i,m}+\gamma_{i,m}$, and
$\tilde{\varphi}_{i}=\varphi_{i}-\sum_{k=1}^{K}\boldsymbol{\beta}_{k,i}^{H}\boldsymbol{\Sigma}_{k}^{(0)}\boldsymbol{\beta}_{k,i}$.
Similarly to the steps from \eqref{eq:stationary:psi:step:compress2a} and \eqref{eq:stationary:psi:step:compress3}, let $\tilde{\omega}_{m}=\sum_{k=1}^{K}\tr\bigl(\boldsymbol{\Xi}_{m,k}^{H}\boldsymbol{\Sigma}_{k}^{(N+M+1)}\boldsymbol{\Xi}_{m,k}\bigr)$,\linebreak
$\tilde{\gamma}_{m,j}=\sum_{k=1}^{K}\tr\bigl(\boldsymbol{\Xi}_{m,k}^{H}\boldsymbol{\Sigma}_{k}^{(j)}\boldsymbol{\Xi}_{m,k}\bigr)$, and $\tilde{\gamma}_{m,s}=\sum_{k=1}^{K}\tr\bigl(\boldsymbol{\Xi}_{m,k}^{H}\boldsymbol{\Sigma}_{k}^{(s)}\boldsymbol{\Xi}_{m,k}\bigr)$, \eqref{eq:stationary:psi:step:compress1_a} can then be  rewritten as
\begin{IEEEeqnarray}{rCl}
&&t\sum_{j=1}^{N}\tilde{\phi}_{m,j}\Delta\eta_{j}+t\sum_{s=1}^{M}\tilde{\phi}_{m,s}\Delta\lambda_{s}
-\lambda_{m}^{-2}\Delta\lambda_{m}+\tilde{\omega}_{m}\Delta\mu_{1}\nonumber\\
&&-tI_{m}\Delta\mu_{2}=t\tilde{\phi}_{m}+tI_{m}\mu_{2}-\lambda_{m}^{-1},\,\forall m=1,\,2,\ldots,\, M\label{eq:stationary:psi:step:compress2b}\;\;\;\;
\end{IEEEeqnarray}
where $\tilde{\phi}_{m,j}=\phi_{m,j}+\tilde{\gamma}_{m,j}$, $\tilde{\phi}_{m,s}=\phi_{m,s}+\tilde{\gamma}_{m,s}$, and $\tilde{\phi}_{m}=\phi_{m}-\sum_{k=1}^{K}\tr\bigl(\boldsymbol{\Xi}_{m,k}^{H}\boldsymbol{\Sigma}_{k}^{(0)}\boldsymbol{\Xi}_{m,k}\bigr)$.
Next, substituting \eqref{eq:block:elimination} into \eqref{eq:feasible:Qk:step} yields
\begin{IEEEeqnarray}{rCl}
\sum_{k=1}^{K}\tr\Bigl(\boldsymbol{\Sigma}_{k}^{(0)}+\sum_{i=1}^{N+M}\Delta\psi_{i}\boldsymbol{\Sigma}_{k}^{(i)}+&&\Delta\mu_{1}\boldsymbol{\Sigma}_{k}^{(N+M+1)}\Bigr)\nonumber\\
&&=P-\sum_{k=1}^{K}\tr(\mathbf{Q}_{k})
\end{IEEEeqnarray}
or equivalently
\begin{equation}
\sum_{i=1}^{N+M}\chi_{i}\Delta\psi_{i}+\chi_{N+M+1}\Delta\mu_{1}=P-\sum_{k=1}^{K}\tr(\mathbf{Q}_{k}+\boldsymbol{\Sigma}_{k}^{(0)}).\label{eq:feasible:Qk:step:transformed}
\end{equation}
where $\chi_{N+M+1}=\sum_{k=1}^{K}\tr(\boldsymbol{\Sigma}_{k}^{(N+M+1)})$ and $\chi_{i}=\sum_{k=1}^{K}\tr(\boldsymbol{\Sigma}_{k}^{(i)})$.
Let us define $\Delta\mathbf{x}=[\Delta\boldsymbol{\psi}^{T}\;\Delta\mu_{1}\;\Delta\mu_{2}]^{T}$. Then, we can stack  \eqref{eq:feasible:psi:step},  \eqref{eq:stationary:psi:step:compress3}, \eqref{eq:stationary:psi:step:compress2b}, and \eqref{eq:feasible:Qk:step:transformed} into a system of linear equations as
\begin{equation}
\mathbf{A}\Delta\mathbf{x}=\mathbf{b}\label{eq:dual:step}
\end{equation}
where $b_{i}=t\tilde{\varphi}_{i}+tP_{i}\mu_{2}-\psi_{i}^{-1}$ for
$i=1,\,2,\,\ldots,\, N$,  $b_{i}=t\tilde{\phi}_{m}+tI_{m}\mu_{2}-\psi_{i}^{-1}$ for
$i=N+1,\,N+2,\,\ldots,\, N+M$ corresponding to $m=1, 2, \ldots, M$, $b_{N+M+1}=P-\sum_{k=1}^{K}\tr(\mathbf{Q}_{k}+\boldsymbol{\Sigma}_{k}^{(0)})$,
and $b_{N+M+2}=P-\mathbf{p}^{T}\boldsymbol{\psi}$. In summary,  the entries of $\mathbf{A}\in\mathbb{C}^{(N+M+2)\times(N+M+2)}$
are given by
\begin{equation}\small
A_{i,j}=\\ \begin{cases}
t\tilde{\varphi}_{i,j}-\frac{\delta_{i,j}}{\psi_{i}^{2}}\negmedspace & 1\leq i,j\leq N\\
t\tilde{\varphi}_{i,m} \negmedspace& 1\leq i\leq N, N+1\leq j\leq N+M\\
t\tilde{\phi}_{m,j} \negmedspace& N+1\leq i\leq N+M, 1\leq j\leq N\\
t\tilde{\phi}_{m,s}-\frac{\delta_{i,j}}{\psi_{i}^{2}} \negmedspace& N+1\leq i, j\leq N+M\\
t\omega_{i} \negmedspace& 1\leq i\leq N,j=N+M+1\\
t\tilde{\omega}_{m} \negmedspace& N+1\leq i\leq N+M,j=N+M+1\\
-tP_{i} \negmedspace& 1\leq i\leq N+M,j=N+M+2\\
\chi_{j} \negmedspace& i=N+M+1,1\leq j\leq N+M+1\\
P_{j} \negmedspace& i=N+M+2,1\leq j\leq N+M\\
0 \negmedspace& \textrm{otherwise}\nonumber
\end{cases}
\end{equation}
where $\delta_{i,j}$ denotes the Kronecker's function, i.e., $\delta_{i,j}=1$
if $i=j$ and $\delta_{i,j}=0$ otherwise, and $\mathbf{p}=[\mathbf{\bar{p}}^T\quad \boldsymbol{\bar{I}}^T]^T=[P_1,\cdots, P_j,\cdots, P_{N+M}]^T$.
 Herein, the complexity of computing the inverse of the KKT matrix in \eqref{eq:dual:step} is of the order $\mathcal{O}\bigl((N+M)^3\bigr)$.

We summarize the proposed algorithm based on barrier method to solve \eqref{eq:BD:Sumrate:dualMAC} in Algorithm 1. In line 7 of Algorithm 1, $r(\{\mathbf{Q}_{k}\},\boldsymbol{\psi},\{\mu_{i}\})$ denotes the residual norm of $\{\mathbf{Q}_k\}$, $\boldsymbol{\psi}$, and $\{\mu_i\}$, which is used in the backtracking line search procedure and is defined as \cite{Tran:MIMO:13}
\begin{equation}\begin{aligned}
&r\bigl(\{\mathbf{Q}_k\}, \boldsymbol{\psi}, \{\mu_i\}\bigr)=\sum_{k=1}^{K}\|\dot{\mathbf{H}}_{k}+\frac{1}{t}\mathbf{Q}_{k}^{-1}-\mu_{1}\mathbf{I}\|_{F}\\
 &\quad+\|\boldsymbol{u}\|_{2}+\|\boldsymbol{w}\|_{2}+|P-\sum_{k=1}^{K}\tr(\mathbf{Q}_{k})|+|P-\mathbf{p}^{T}\boldsymbol{\psi}|
\end{aligned}\end{equation}
where $u_i=\sum_{k=1}^{K}\mathbf{v}_{k,i}^{H}\Bigl[\boldsymbol{\Pi}_k^{-1}-(\bar{\mathbf{V}}_{k}^{H}\boldsymbol{\Lambda}\bar{\mathbf{V}}_{k})^{-1}\Bigr]\mathbf{v}_{k,i}-\frac{1}{t}\eta_{i}^{-1}+\mu_{2}P_{i}$ for $i=1,\, 2,\, \cdots, N$, and
$w_m=\sum_{k=1}^{K}\tr\Bigl(\tilde{\mathbf{G}}_{mk}\Bigl[\boldsymbol{\Pi}_k^{-1}-(\bar{\mathbf{V}}_{k}^{H}\boldsymbol{\Lambda}\bar{\mathbf{V}}_{k})^{-1}\Bigr]\tilde{\mathbf{G}}_{mk}^{H}\Bigr)-\frac{1}{t}\lambda_{m}^{-1}+\mu_{2}I_{m}$ for $m=1,\, 2,\, \cdots, M$. The backtracking line search stops when the residual norm is smaller than a  predetermined threshold, i.e., $\epsilon$ as shown in line 11.

\begin{algorithm}[t]
\begin{algorithmic}[1]

\protect\caption{The proposed numerical algorithm to solve \eqref{eq:BD:Sumrate:dualMAC}}

\label{algo:proposed:BD:PAPC}

\global\long\def\algorithmicrequire{\textbf{Initinalization:}}

\REQUIRE $\mathbf{Q}_{k}:=\mathbf{I}_{n_{k}}$, $\boldsymbol{\psi}:=\boldsymbol{1}$,
$\mu_{1}=\mu_{2}=1$, $t:=t_{0}$,  $\gamma$, and tolerance $\epsilon>0$

\REPEAT[Outer iteration]

 \REPEAT[inner iteration (centering step)]

\STATE Solve \eqref{eq:discrete:Sylvester} to find $\boldsymbol{\Sigma}_{k}^{(i)}$
for $1\leq k\leq K$ and $0\leq i\leq N+M+1$.

\STATE Solve \eqref{eq:dual:step} to find $\Delta\boldsymbol{\psi}$
and $\Delta\mu_{i}$.

\renewcommand{\algorithmicbody}{\text{Backtracking line search:}}

\BODY

\STATE $s=1$

\WHILE{ $r(\{\mathbf{Q}_{k}\}+s\{\Delta\mathbf{Q}_{k}\},\boldsymbol{\psi}+s\Delta\boldsymbol{\psi}, \{\mu_{i}\}+s\{\Delta\mu_{i}\})>(1-\alpha s)r(\{\mathbf{Q}_{k}\},\boldsymbol{\psi},\{\mu_{i}\})$ $\text{or}$ $\{\mathbf{Q}_{k}\}+s\{\Delta\mathbf{Q}_{k}\}\preceq\mathbf{0}$}

\STATE $s=\beta s$

\ENDWHILE \ENDBODY

\STATE Update primal and dual variables: $\mathbf{Q}_{k}:=\mathbf{Q}_{k}+s\Delta\mathbf{Q}_{k}$;
$\boldsymbol{\psi}:=\boldsymbol{\psi}+s\Delta\boldsymbol{\psi}$, $\mu_i:=\mu_i+s\Delta\mu_i$

\UNTIL $r(\{\mathbf{Q}_{k}\},\boldsymbol{\psi}, \mu_{i})<\epsilon$

\STATE Increase $t$: $t=\gamma t$.

\UNTIL  $t$ is sufficiently large to tolerate the duality gap.

\end{algorithmic}
\end{algorithm}

\subsection*{Convergence and Complexity Analysis}
Algorithm \ref{algo:proposed:BD:PAPC} which is based on a barrier method is guaranteed to converge to a solution to the convex-concave problem in \eqref{eq:BD:Sumrate:dualMAC}, following the same arguments as those in \cite[Sec.~11.3.3]{Stephen}. Moreover, the numerical results provided in Fig. \ref{fig:Convergencebehavior:Iteration:N} demonstrate that Algorithm 1 exhibits a superlinear convergence rate, i.e., it converges very fast when approaching the optimal solution. The complexity of Algorithm \ref{algo:proposed:BD:PAPC} is mostly due to solving \eqref{eq:discrete:Sylvester} and \eqref{eq:dual:step}, which require  complexity  $\mathcal{O}(n_{k}^{3})$ and $\mathcal{O}\bigl((N+M)^3\bigr)$, respectively. That is to say, the proposed algorithm requires significantly lower complexity, compared to a generic method that has complexity of $\mathcal{O}(K^3N^6)$ as mentioned previously. We note that semidefinite programing (SDP) can be applied to solve the relaxed problem of \eqref{eq:BD:Sumrate:PAPC:relaxed} since the logdet function can be represented by semidefinite cone \cite[page~149]{R1}. Modern SDP solvers are usually based on a specific interior point method which is called the primal-dual path following method. However, such a method, e.g., the one in \cite{R2}, will have a per iteration complexity of $\mathcal{O}(K^4\bar{n}_k^4)$ which is higher than that in our proposed algorithm, especially when $K$ is large.

\section{Numerical Results}\label{Numericalresults}

In this section, we provide numerical examples to illustrate the results of the proposed algorithm. The entry of the channel matrix  from the secondary BS to the $k$th SU is modeled as correlated Rayleigh fading, i.e., $\mathbf{H}_k=\mathbf{P}_{k}^{1/2}\ddot{\mathbf{H}}_k\mathbf{R}_{k}^{1/2}$ where  $\ddot{\mathbf{H}}_k$ is a matrix of independent circularly symmetric complex Gaussian (CSCG) random variables  with zero mean and unit variance, and $\mathbf{P}_{k}\in\mathbb{C}^{n_k\times n_k}$ and $\mathbf{R}_{k}\in\mathbb{C}^{N\times N}$ are the receive and transmit correlation matrices, respectively.  In our simulation setups, the exponential correlation model is used, where $\mathbf{P}_{k}$ and $\mathbf{R}_{k}$ are generated as $[\mathbf{P}_{k}]_{i,j}=r_{k}^{|i-j|}$
and $[\mathbf{R}_{k}]_{i,j}=\tilde{r}_{k}^{|i-j|}$, respectively \cite{Poyka,Tran:TSP:12}. The correlation coefficients $r_{k}$ and $\tilde{r}_{k}$
are given by $r_{k}=r\times e^{j\breve{\phi}_{k}}$ and $\tilde{r}_{k}=r\times e^{j\hat{\phi}_{k}}$, where $r\in[0,1]$ and $\breve{\phi}_{k}$ and $\hat{\phi}_{k}$ are i.i.d. and uniformly distributed over $[0, 2\pi)$. The channel from the secondary BS to the $m$th PU is  generated  as $\mathbf{G}_m=\mathbf{P}_{m}^{1/2}\ddot{\mathbf{G}}_m\mathbf{R}_{m}^{1/2}$ and the covariance matrices $\mathbf{P}_m$ and $\mathbf{R}_m$ are generated similarly as described above. The number of antennas at each SU receiver is set to $n_k = 2$ for all $k$. For simplicity, we further assume the interference thresholds for all PUs are equal, i.e., $I_m = I$, for all $m$,  and the resultant power constraint for each antenna  is $P_n = P/N$ for all $n$, where $P$ is the total power budget. The  SRs  are averaged over 10,000 simulation trials. For the purpose of exposition,   the results in  Figs. \ref{fig:sumrate:transmitpower}-\ref{fig:sumrate:transmitpower:PUsystem} are shown with $r=0$, i.e.,  uncorrelated Rayleigh fading.

\begin{figure}
    \begin{center}
    \begin{subfigure}[Convergence results of Algorithm 1 for different numbers of transmit antennas at the secondary BS. The number of PUs is $M=1$.]{
        \includegraphics[width=0.4\textwidth]{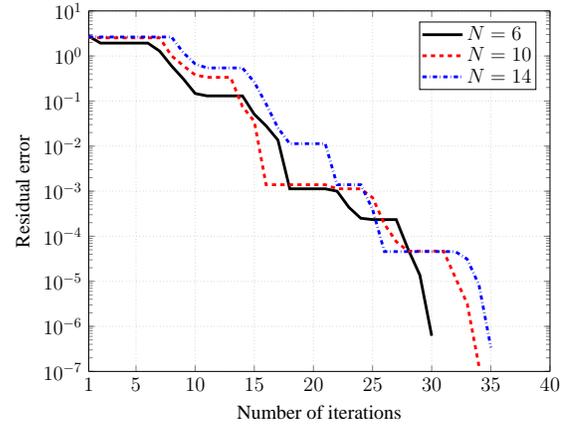}}
    		\label{fig:a}
				\end{subfigure}
		\hfill \\
		 \begin{subfigure}[Convergence results of Algorithm 1 for different numbers of PUs. The number of transmit antennas is $N=10.$]{
        \includegraphics[width=0.4\textwidth]{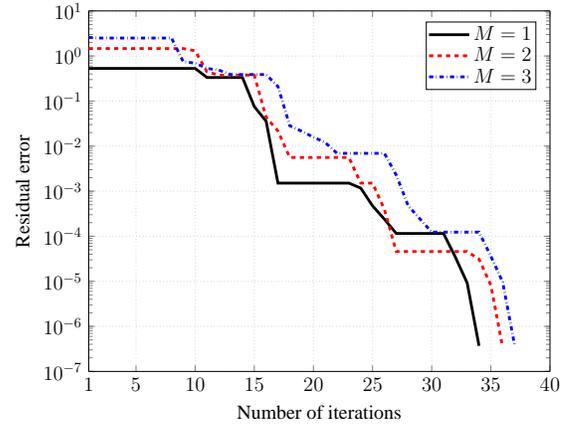}}
        \label{fig:b}
    \end{subfigure}
	  \caption{Convergence results of the proposed algorithm (a) for  different numbers of transmit antennas at the secondary BS, and (b) for different numbers of PUs. Each curve is obtained for one channel realization. The parameters of Algorithm 1 are  as follows. The tolerance  is set to  $\epsilon = 10^{-5}$. The barrier parameters $\gamma$ and $t_0$ are  set to 1 and 50, respectively. The backtracking line search parameters in Algorithm 1 are set to  $\alpha = 0.01$ and $\beta = 0.5$. In this example, we set the network parameters as $K = 2, n_k = 2, \forall k,	\tilde{n}_m = 2, \forall m, P = 10\, \mbox{dB},$ and $ I = 5\, \mbox{dB}$.}\label{fig:Convergencebehavior:Iteration:N}
\end{center}
\end{figure}

\begin{figure}[t]
\centering
\includegraphics[width=0.45\textwidth]{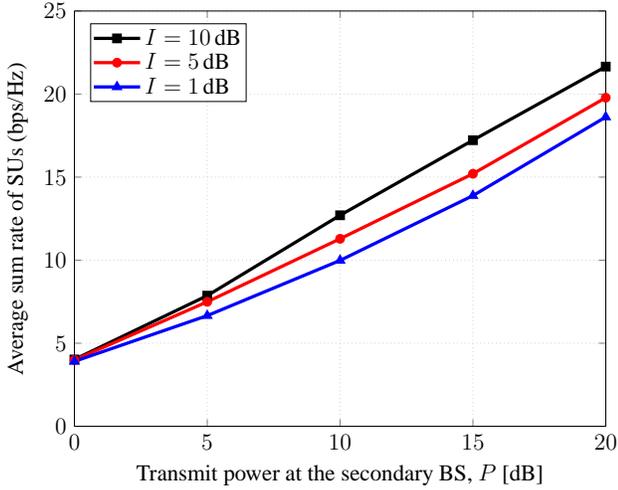}
\caption{Average sum rate of SUs with the total transmit power for different interference thresholds at PUs. The network configuration is $N=10,\, K=3,\, M=2,\, n_k=2,\, \forall k,\,$ and $ \tilde{n}_m=2,\;\,\forall m$.}
\label{fig:sumrate:transmitpower}
\end{figure}

\begin{figure}[t]
\centering
\includegraphics[width=0.45\textwidth]{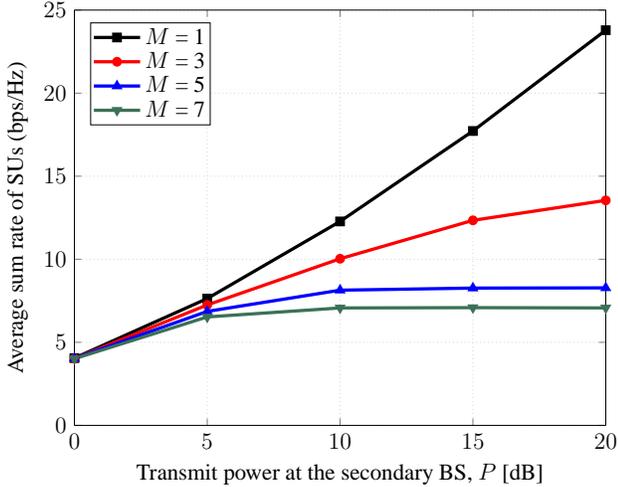}
\caption{Average sum rate of SUs with the total transmit power for different numbers of PUs. In this example, we set the network parameters as $K = 3, n_k=2,\;\forall k$, $\tilde{n}_m=2, \;\forall m, I = 5\, \mbox{dB},$ and $ N=10$.}
\label{fig:sumrate:transmitpower:PU}
\end{figure}

\begin{figure}[t]
\centering
\includegraphics[width=0.45\textwidth]{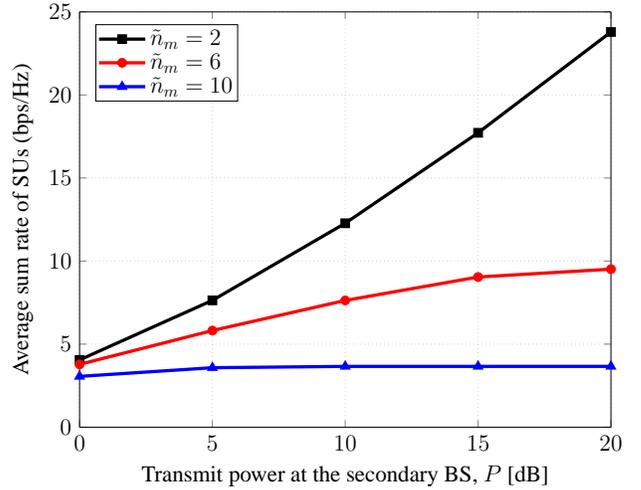}
\caption{Average sum rate of SUs with the total transmit power for different numbers of antennas at PUs. The network parameters are set to $K = 3, n_k=2,\;\forall k$, $M = 1, I = 5\, \mbox{dB},$ and $ N=10$.}
\label{fig:sumrate:transmitpower:n-m}
\end{figure}

\begin{figure}
    \begin{center}
    \begin{subfigure}[Average sum rate of SUs versus the number of transmit antennas at the secondary BS. The transmit power  is set to $P = 10$ dB.]{
        \includegraphics[width=0.45\textwidth]{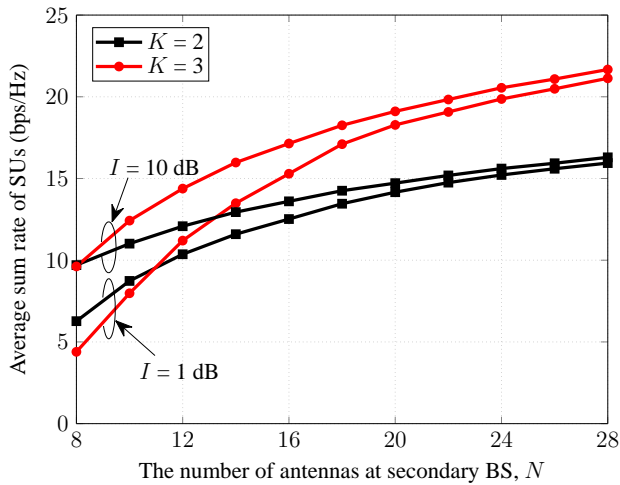}}
    		\label{fig:4:a}
				\end{subfigure}
		\hfill \\
		
		 \begin{subfigure}[Average sum rate of SUs versus the transmit power at the secondary BS. The interference threshold  is set to $I = 1$ dB.]{
        \includegraphics[width=0.45\textwidth]{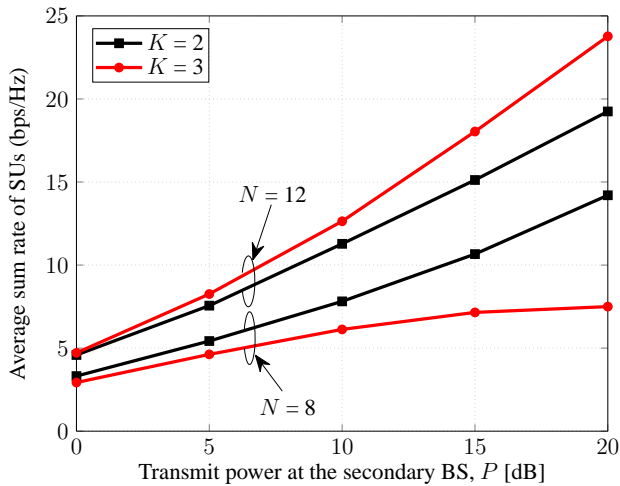}}
        \label{fig:4:b}
    \end{subfigure}
	  \caption{Average sum rate of SUs, (a) versus the number of transmit antennas, and (b) versus the transmit power at the secondary BS. In the cognitive network, we choose $M = 2, \tilde{n}_m=2, \;\forall m$, and $n_k=2,\;\forall k$.}\label{fig:sumrate:transmitpower:N}
\end{center}
\end{figure}

\begin{figure}[t]
\centering
\includegraphics[width=0.45\textwidth]{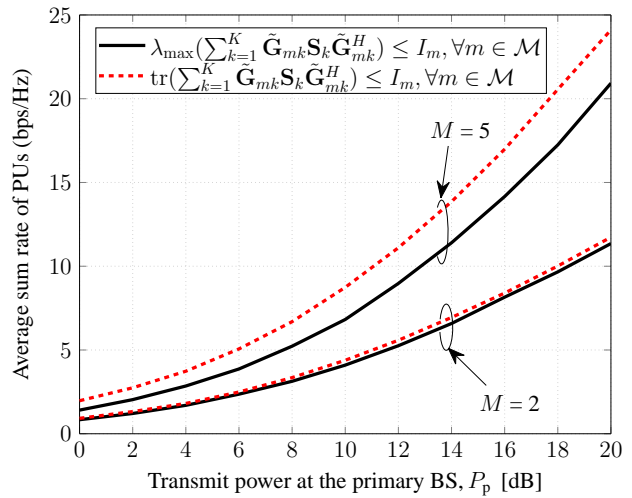}
\caption{Average sum rate of PUs with the total transmit power at the primary BS, $P_{\text{p}}$. For both the  systems, we set network parameters as $\tilde{N} = N = 10$, $K =3$, $\tilde{n}_m = 2,\,\forall m$, $n_k = 2,\,\forall k$, $I$ = 1 dB, and $P$ = 10 dB. }
\label{fig:sumrate:transmitpower:PUsystem}
\end{figure}

\begin{figure}[t]
\centering
\includegraphics[width=0.45\textwidth]{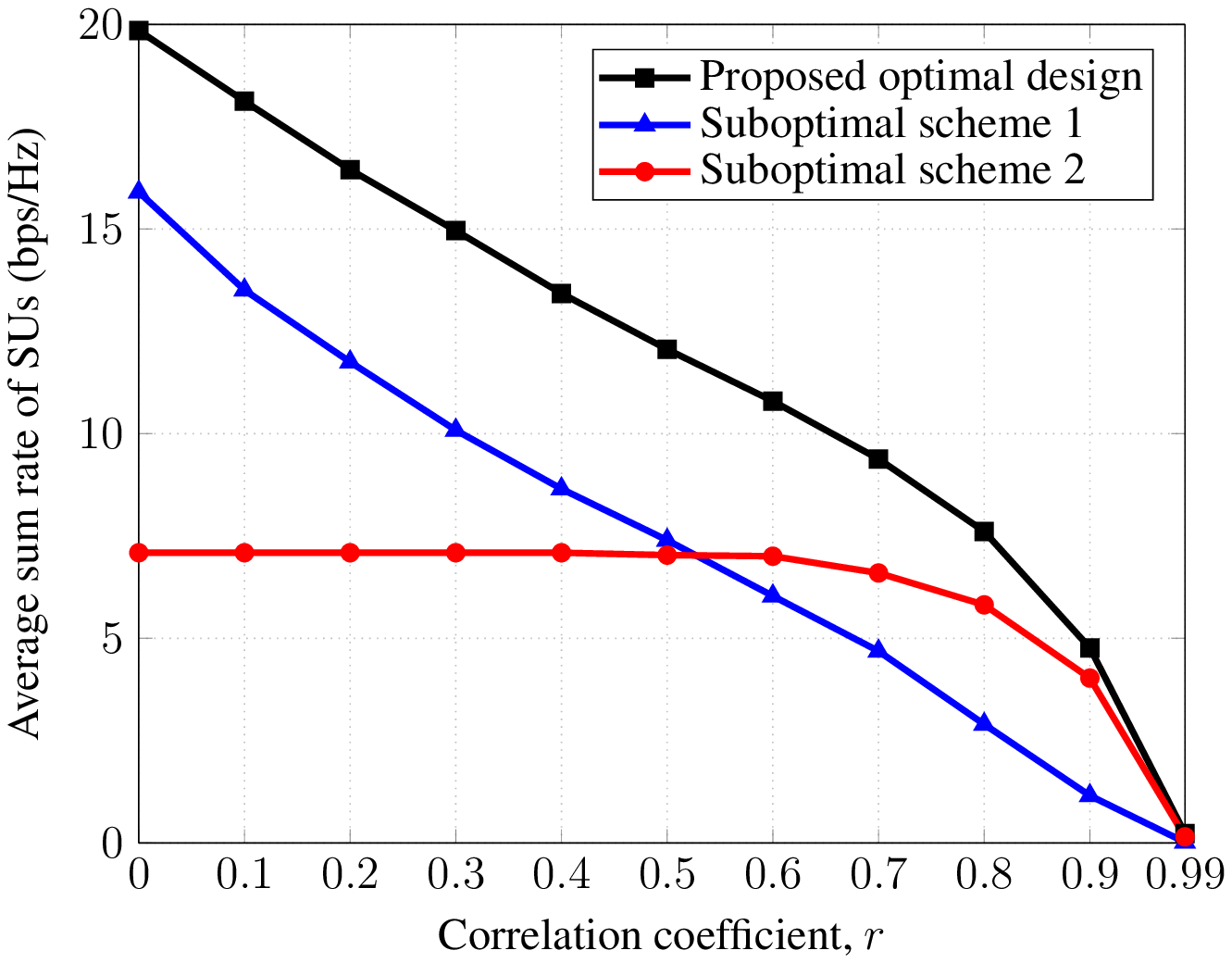}
\caption{Average sum rate of SUs as a function of the correlation coefficient, $r$. The network parameters are set to $ N = 10$, $K =3$, $n_k = 2,\,\forall k$, $M = 2$, $\tilde{n}_m = 2,\,\forall m$, $I$ = 5 dB, and $P$ = 20 dB. }
\label{fig:sumrate:correlation}
\end{figure}

In Fig. \ref{fig:Convergencebehavior:Iteration:N} we show the convergence rate of the proposed algorithm for the simulation settings given in the caption.    In particular, we plot the convergence rate of the proposed barrier method for different numbers of transmit antennas at the secondary BS in Fig. 2(a),   and  for different numbers of PUs in Fig. 2(b). The initial values for the primal and dual variables in Algorithm 1 are randomly generated. We can see that Algorithm 1 shows a very fast convergence rate when it is approaching the optimal solution. We note that this convergence result can  be expected for an algorithm based on Newton's method. We can also see that its convergence rate is slightly sensitive to the network configurations.

Fig. \ref{fig:sumrate:transmitpower} shows the impact of the interference thresholds on the SR of the secondary system with respect to the transmit power constraint.  As can be seen from Fig.~\ref{fig:sumrate:transmitpower}, decreasing the interference threshold $I$  degrades the system performance. The performance gains achieved  for higher interference threshold are due to the fact that more transmit power can be used when the interference threshold constraints are set to a higher value. Specifically, the SRs of the secondary system with  different levels of the interference threshold constraint are of the same order in the low transmit power regime. The reason is that the considered interference thresholds still allow the secondary system to operate with nearly full transmit power.

Fig. \ref{fig:sumrate:transmitpower:PU} depicts the impact of the number of PUs on the performance of the cognitive network.  We can see that the average SR of secondary system is degraded as the number of PUs $M$ increases, and the degradation is dramatic for large transmit power at the secondary BS. For low power regime, the secondary BS mainly focuses on maximizing the SR of the secondary system and pays little attention to the primary system since the interference constraints are likely satisfied for all PUs. Accordingly, the SRs are quite similar in such cases for different numbers of PUs as observed in Fig. \ref{fig:sumrate:transmitpower:PU}. As the transmit power becomes sufficiently large, the BS needs to avoid transmitting its signals over the spatial space of PUs. In this case, the degree of freedom left for the secondary system is reduced if the number of PUs increases, and this explains  the noticeable differences of SRs for large numbers of PUs.

In Fig. \ref{fig:sumrate:transmitpower:n-m}, we consider the case of a single PU (i.e., $M = 1$) and examine the impact of the number of antennas at the PU on the performance of the secondary system. All the remaining parameters are set to the same values as those in Fig. \ref{fig:sumrate:transmitpower:PU}. Similarly, it is observed that increasing the number of antennas at the PU severely deteriorates the average SR of the secondary performance, especially for large transmit power at the BS. We can intuitively explain this observation by simply treating each transmit antenna of the PU as a single PU and recalling the discussions presented for Fig. \ref{fig:sumrate:transmitpower:PU}.

In the next numerical example, we plot the average SR versus the number of transmit antennas and versus the transmit power at the secondary BS  in Fig. \ref{fig:sumrate:transmitpower:N}(a) and Fig. \ref{fig:sumrate:transmitpower:N}(b), respectively. As expected, the average SR  improves as the number of transmit antennas increases since more degrees of freedom are added to the secondary system. A large interference threshold (i.e., $I$ = 10 dB in Fig. \ref{fig:sumrate:transmitpower:N}(a)) allows the secondary system to fully exploit the available degrees of freedom. In such a case, more SUs result in more multiuser diversity in the secondary system, which accordingly increases the average SR.  On the other hand, for a small interference threshold in both the Fig. \ref{fig:sumrate:transmitpower:N}(a) and Fig. \ref{fig:sumrate:transmitpower:N}(b), the  average SR  first decreases for small $N$ and then increases for large $N$ when more SUs are served. Another interesting point is that the gap between the curves  is reduced, as the number of antennas at the secondary BS increases. We recall that for small $N$ and small interference threshold, the secondary system lacks  degree of freedom for leveraging multiuser diversity.

As  mentioned earlier, interference from the secondary BS to a PU   is a matrix which is non-degrading and there are possibly several ways  to control the interference. Thus it is very interesting to investigate how different types of interference constraints affect the performance of the primary system.  In this numerical experiment, we consider a primary system where the primary BS is equipped with $\tilde{N}$ antennas and employs ZF precoding with PAPCs \cite{Tran:MIMO:13}. The interference from the secondary system is simply treated as background thermal noise at the PUs. Under this setup, we plot the average achieved SR of the primary system for two different types of interference constraints, one considered in this paper shown in \eqref{eq:BD:Sumrate:PAPC:relaxed_b} and the other obtained by replacing \eqref{eq:BD:Sumrate:PAPC:relaxed_b} with $\lambda_{\max}(\sum_{k=1}^{K}\tilde{\mathbf{G}}_{mk} \mathbf{S}_k \tilde{\mathbf{G}}_{mk}^H) \leq I_m, \forall m \in \mathcal{M}$. In other words, we impose that the largest eigenvalue of the interference matrix should be smaller than a predetermined threshold. As can be seen from  Fig. \ref{fig:sumrate:transmitpower:PUsystem}, the primary system achieves a higher SR with the sum of all eigenvalues  than with the largest eigenvalue  of the interference matrix. This implies that the  new type of interference constraint is stricter than our considered one. The gaps between the two schemes are negligible for small number of receive antennas at the PUs (i.e., $M = 2$ in Fig. \ref{fig:sumrate:transmitpower:PUsystem}). The reason is that in such cases the number of eigenvalues in the interference matrix is small and thus the trace of the interference matrix is mostly contributed by the largest eigenvalue. However, when $M$  becomes large, the impact of smaller eigenvalues of the interfering matrix can be comparable to the largest one. Thus, imposing a constraint on the maximum eigenvalue may not appropriately characterize the interference situation. In other words, the interference can be more severe that it is supposed to be. This leads to a reduction on the achieved SR of the primary system as shown in Fig. \ref{fig:sumrate:transmitpower:PUsystem}.

Finally, we study the SR of SUs as a function of the correlation coefficient $r$ in Fig.~\ref{fig:sumrate:correlation}. In particular, we  compare the performance of the proposed design with that of two suboptimal schemes. For suboptimal scheme 1, the ZF precoding at the secondary BS is chosen to eliminate the interference at both SUs and PUs. For suboptimal scheme 2, the ZF precoding for the $k$th SU is given by $\mathbf{T}_k = \bar{\mathbf{V}}_k\mathbf{\dot{V}}_k\boldsymbol{\Phi}^{1/2}_k$ where $\mathbf{\dot{V}}_k$ contains the $n_k$ singular vectors of $\tilde{\mathbf{H}}_k$, i.e., it comes from a compact SVD of $\tilde{\mathbf{H}}_k$: $\tilde{\mathbf{H}}_k = \mathbf{\dot{U}}_k\mathbf{\dot{D}}_k\mathbf{\dot{V}}_k^H$, and $\boldsymbol{\Phi}_k\in\mathbb{C}^{n_k\times n_k}$ is the solution of \eqref{eq:BD:Sumrate:PAPC:relaxed} \cite{Tran:MIMO:13}. The results show that the SR of all the schemes decreases as the correlation coefficient increases.  This is because a high  correlation coefficient reduces the spatial diversity of the secondary system. Importantly, the proposed design shows superior performance compared to the others. Another interesting observation is that, in the strong correlation regime, the sum rate of the suboptimal scheme 2 is larger than that of the suboptimal scheme 1  and vice versa in the weak correlation regime. The reason is that as the spatial correlation at the transmitter and receiver is stronger, the channels become more directional and thus interference isolation can be achieved physically. Accordingly, the effect of the secondary system on the primary one will be small.  Consequently, the suboptimal scheme 2 has more degrees of freedom when designing the precoders for the secondary system as it does not force those precoders to be in the null space of the PU channels.

\section{Conclusions}\label{Conclusion}
In this paper, we have considered the sum rate maximization problem for downlink transmission of MIMO CR networks in the presence of
 multiple SUs and  PUs. The design problem is subject to per-antenna power constraints at the secondary BS and interference constraints at the PUs. Adopting ZF precoding technique and using a rank relaxation method which is shown to be tight, we have transformed the problem of SR maximization into the one of finding a saddle point of a convex-concave program. Then a computationally efficient algorithm has been proposed based on a barrier method to solve the resulting convex-concave problem, exploiting its special properties. The proposed algorithm was numerically shown to have a superlinear convergence behavior which is almost independent of the problem size. Through numerical experiments, we have illustrated how the performances of the secondary and the primary systems vary with the type of the interference constraints considered in the paper. In particular, we have shown that the performance of the secondary system degrades significantly and reaches a saturated value when either the number of primary users or the number of antennas at the primary user is large. Also, we have discussed different ways of controlling amount of the interference caused by the secondary system.  Specifically, we have concluded that using the trace of the interference matrix may reflect the interference situation better, compared to using the largest eigenvalue.

\appendix[Proof of Theorem \ref{thm:BC-MAC-duality}]{}
To prove Theorem 1, we will follow the same steps as those in \cite{Tran:MISO:13} but customize them to our considered problem.  Particularly we show  that  \eqref{eq:BD:Sumrate:dualMAC} is the dual problem of the relaxed problem of \eqref{eq:BD:Sumrate:PAPC:relaxed}. Let us start by writing the partial Lagrangian function of the relaxed problem of  \eqref{eq:BD:Sumrate:PAPC:relaxed}  as
\begin{equation}\small\label{eq:larg:expression}
\begin{aligned}
&\mathcal{L}\left(\{\mathbf{S}_k\}, \{\eta_n\}, \{\lambda_m\}\right)=\sum_{k=1}^K\log|\mathbf{I}+\tilde{\mathbf{H}}_k\mathbf{S}_k\tilde{\mathbf{H}}_k^H|\\
 &+\sum_{n=1}^N\eta_n\Bigl(P_n-\negmedspace\sum_{k=1}^K\tr\bigl(\mathbf{S}_k\mathbf{B}_k^{(n)}\bigr)\Bigr)+\sum_{m=1}^M\lambda_m\Bigl(I_m-\tr\bigl(\mathbf{S}_k\Dot{\mathbf{G}}_{mk}\bigr)\Bigr)
\end{aligned}
\end{equation}
where $\mathbf{B}_k^{(n)}\triangleq\bar{\mathbf{D}}_k^H\bar{\mathbf{D}}_k$, $\bar{\mathbf{D}}_k=[\mathbf{0}_{n-1}^T\quad 1 \quad \mathbf{0}_{N-n}^T]\bar{\mathbf{V}}_k$ and $\Dot{\mathbf{G}}_{mk}\triangleq\tilde{\mathbf{G}}_{mk}^H\tilde{\mathbf{G}}_{mk}$. $\{\eta_n\}$ and $\{\lambda_m\}$ are Lagrangian multipliers corresponding to the constraints in \eqref{eq:BD:Sumrate:PAPC:relaxed_a} and  \eqref{eq:BD:Sumrate:PAPC:relaxed_b}, respectively. From the dual problem, the dual objective of \eqref{eq:BD:Sumrate:PAPC:relaxed} is given by
\begin{equation}
\mathcal{D}(\{\eta_n\}, \{\lambda_m\})= \underset{\left\{\mathbf{S}_k\succeq\mathbf{0}\right\}}{\maxi}\mathcal{L}\left(\{\mathbf{S}_k\}, \{\eta_n\}, \{\lambda_m\}\right).
\label{eq:dual}
\end{equation}
To solve the maximization in \eqref{eq:dual} for a given set of Lagrangian multipliers $\left(\{\eta_n\}, \{\lambda_m\}\right)$, we first rewrite the Lagrangian function as
\begin{equation}
\begin{aligned}
&\mathcal{L}\left(\{\mathbf{S}_k\}, \boldsymbol{\eta}, \boldsymbol{\lambda}\right)=\sum_{k=1}^K\log|\mathbf{I}+\tilde{\mathbf{H}}_k\mathbf{S}_k\tilde{\mathbf{H}}_k^H|\\
&\qquad\qquad-\sum_{k=1}^K\tr\bigl(\mathbf{S}_k\mathbf{\Omega}_k\bigr)+\mathbf{\bar{p}}^T\boldsymbol{\eta}+\boldsymbol{\bar{I}}^T\boldsymbol{\lambda}
\end{aligned}\label{eq:Lagragian_rewrite}
\end{equation}
where $\mathbf{\Omega}_k \triangleq \sum_{n=1}^N\eta_n\mathbf{B}_{k}^{(n)}+\sum_{m=1}^M\lambda_m\Dot{\mathbf{G}}_{mk}=\mathbf{\bar{V}}_k^H\left(\diag(\boldsymbol{\eta})+\sum_{m=1}^M\lambda_m\mathbf{G}_{m}^H\mathbf{G}_{m}\right)\mathbf{\bar{V}}_k$, $\mathbf{\bar{p}}=[P_1, P_2,\cdots, P_N]^T$, $\boldsymbol{\eta} = [\eta_1, \eta_2,\cdots, \eta_N]^T$, $\boldsymbol{\bar{I}}=[I_1, I_2, \cdots, I_M]^T$, and $\boldsymbol{\lambda}=[\lambda_1, \lambda_2,\cdots, \lambda_M]^T$.
Since $\mathbf{\Omega}_k$ is invertible\footnote{It implies that $\{\eta_n\}$ or $\{\lambda_m\}$  must be positive. This is because of all power at the secondary BS should be used up or the interference constraint should be met at the optimum, i.e., $\{\eta_n^*\}>0$ or $\{\lambda_m^*\}>0$.}, let $\mathbf{\bar{S}}_k=\mathbf{\Omega}_k^{1/2}\mathbf{S}_k\mathbf{\Omega}_k^{1/2}$. The maximization in \eqref{eq:Lagragian_rewrite} is rewritten as
\begin{equation}
\begin{aligned}
\mathcal{L}\left(\{\mathbf{\bar{S}}_k\}, \boldsymbol{\eta}, \boldsymbol{\lambda}\right)&=\sum_{k=1}^K\log|\mathbf{I}+\tilde{\mathbf{H}}_k\mathbf{\Omega}_k^{-1/2}\mathbf{\bar{S}}_k\mathbf{\Omega}_k^{-1/2}\tilde{\mathbf{H}}_k^H|\\
                                                                &\quad-\sum_{k=1}^K\tr\bigl(\mathbf{\bar{S}}_k\bigr)+\mathbf{\bar{p}}^T\boldsymbol{\eta}+\boldsymbol{\bar{I}}^T\boldsymbol{\lambda}.
\end{aligned}\label{eq:Lagragian_rivised}
\end{equation}
To compute the dual objective, let $\mathbf{U}_{k}\mathbf{D}_{k}\mathbf{V}_{k}^{H}$ be a SVD of $\tilde{\mathbf{H}}_{k}\boldsymbol{\Omega}_{k}^{-1/2}$ where $\mathbf{D}_{k}$ is square and diagonal. Then, we can apply a result  in \cite[Appendix A]{Vishwanath} to find $\mathcal{D}(\boldsymbol{\eta}, \boldsymbol{\lambda})$ as
\begin{IEEEeqnarray}{rCl}
\mathcal{D}(\boldsymbol{\eta}, \boldsymbol{\lambda})&= \underset{\{\mathbf{Q}_k\}\succeq\mathbf{0}}{\maxi}\sum_{k=1}^K\log|\mathbf{I}+\mathbf{\Omega}_k^{-1/2}\tilde{\mathbf{H}}_k^H\mathbf{Q}_k\tilde{\mathbf{H}}_k\mathbf{\Omega}_k^{-1/2}|\nonumber\\
                                                               &\quad -\sum_{k=1}^K\tr\bigl(\mathbf{Q}_k\bigr)+\mathbf{\bar{p}}^T\boldsymbol{\eta}+\boldsymbol{\bar{I}}^T\boldsymbol{\lambda}
\label{eq:dual_equivalent}
\end{IEEEeqnarray}
where the relationship between $\mathbf{\bar{S}}_k$ and $\mathbf{Q}_k$ is given by
\begin{equation}\begin{aligned}
&\mathbf{\bar{S}}_k=\mathbf{V}_k\mathbf{U}_k^H\mathbf{Q}_k\mathbf{U}_k\mathbf{V}_k^H\\
&\mathbf{Q}_k=\mathbf{U}_k\mathbf{V}_k^H\mathbf{\bar{S}}_k\mathbf{V}_k\mathbf{U}_k^H.
\end{aligned}\label{eq:relation:SQ}
\end{equation}
From \eqref{eq:dual_equivalent} we can further write $\mathcal{D}(\boldsymbol{\eta}, \boldsymbol{\lambda})$ as
\begin{IEEEeqnarray}{rCl}
\mathcal{D}(\boldsymbol{\psi})&&= \underset{\{\mathbf{Q}_k\}\succeq\mathbf{0}}{\maxi}\sum_{k=1}^K\log|\mathbf{I}+\mathbf{\Omega}_k^{-1/2}\tilde{\mathbf{H}}_k^H\mathbf{Q}_k\tilde{\mathbf{H}}_k\mathbf{\Omega}_k^{-1/2}|\nonumber\\
                              &&\qquad                                  -\sum_{k=1}^K\tr\bigl(\mathbf{Q}_k\bigr)+\mathbf{p}^T\boldsymbol{\psi}\nonumber\\
															&&= \underset{\{\mathbf{Q}_k\}\succeq\mathbf{0}}{\maxi}\sum_{k=1}^K\log\frac{|\mathbf{\bar{V}}_k^H\mathbf{\Lambda}\mathbf{\bar{V}}_k+\tilde{\mathbf{H}}_k^H\mathbf{Q}_k\tilde{\mathbf{H}}_k|}{|\mathbf{\bar{V}}_k^H\mathbf{\Lambda}\mathbf{\bar{V}}_k|}\nonumber\\
															&&\qquad-\sum_{k=1}^K\tr\bigl(\mathbf{Q}_k\bigr)+\mathbf{p}^T\boldsymbol{\psi}													
\label{eq:dual_rewrite}
\end{IEEEeqnarray}
where $\mathbf{\Lambda}=\diag(\boldsymbol{\eta})+\sum_{m=1}^M\lambda_m\mathbf{G}_{m}^H\mathbf{G}_{m}$, $\mathbf{p}=[\mathbf{\bar{p}}^T\quad \boldsymbol{\bar{I}}^T]^T$, and $\boldsymbol{\psi}=[\boldsymbol{\eta}^T\quad\boldsymbol{\lambda}^T]^T$.

Now the dual problem of \eqref{eq:BD:Sumrate:PAPC:relaxed} is given by
\begin{equation}\begin{aligned}
&\mini\mathcal{D}(\boldsymbol{\psi})\\
&\st \boldsymbol{\psi}\geq 0.
\end{aligned}\label{eq:appendix:38}\end{equation}
or equivalently
\begin{equation}
\begin{array}{rl}
&\underset{\boldsymbol{\psi}\geq 0}{\mini}\;\underset{\{\mathbf{Q}_{k}\}\succeq\mathbf{0}}{\maxi}  \sum_{k=1}^K\log\frac{|\mathbf{\bar{V}}_k^H\mathbf{\Lambda}\mathbf{\bar{V}}_k+\tilde{\mathbf{H}}_k^H\mathbf{Q}_k\tilde{\mathbf{H}}_k|}{|\mathbf{\bar{V}}_k^H\mathbf{\Lambda}\mathbf{\bar{V}}_k|} \\
&\qquad\qquad\qquad-\sum_{k=1}^K\tr\bigl(\mathbf{Q}_k\bigr)+\mathbf{p}^T\boldsymbol{\psi}
\end{array}\label{eq:BD:Sumrate:dualMAC:proof}
\end{equation}
To arrive at a minimax program given in \eqref{eq:BD:Sumrate:dualMAC}, we introduce  another optimization variable $\vartheta\geq 0$ and rewrite \eqref{eq:BD:Sumrate:dualMAC:proof}  as
\begin{equation}
\begin{array}{rl}
\underset{\boldsymbol{\psi}\geq 0}{\mini}\;\underset{\vartheta\geq0, \{\mathbf{Q}_{k}\}\succeq\mathbf{0}}{\maxi} & \sum_{k=1}^K\log\frac{|\mathbf{\bar{V}}_k^H\mathbf{\Lambda}\mathbf{\bar{V}}_k+\tilde{\mathbf{H}}_k^H\mathbf{Q}_k\tilde{\mathbf{H}}_k|}{|\mathbf{\bar{V}}_k^H\mathbf{\Lambda}\mathbf{\bar{V}}_k|}\\
                                           &\qquad -\vartheta P+\mathbf{p}^T\boldsymbol{\psi}\\
\st   &  \sum_{k=1}^K\tr\bigl(\mathbf{Q}_k\bigr)\leq \vartheta P.
\end{array}\label{eq:BD:Sumrate:dualMAC:proof:new}
\end{equation}
It is easy to see that \eqref{eq:BD:Sumrate:dualMAC:proof:new} is equivalent to \eqref{eq:BD:Sumrate:dualMAC:proof} since the inequality (use subequation environment and give a label to the trace($\cdot$) constraint) in \eqref{eq:BD:Sumrate:dualMAC:proof:new} must hold with equality at optimality; otherwise we can scale down $\vartheta$ to achieve a strictly larger objective. Next we make a change of variables as
\begin{equation}\begin{aligned}
&\boldsymbol{\tilde{\eta}}=\boldsymbol{\eta}/\vartheta\\
&\boldsymbol{\tilde{\lambda}}=\boldsymbol{\lambda}/\vartheta\\
&\mathbf{\tilde{Q}}_k=\mathbf{Q}_k/\vartheta.
\end{aligned}\label{eq:newvariable}
\end{equation}
and define
\begin{equation}\begin{aligned}
&\boldsymbol{\tilde{\psi}}=\boldsymbol{\psi}/\vartheta=[\boldsymbol{\tilde{\eta}}^T\quad\boldsymbol{\tilde{\lambda}}^T]^T\\
&\mathbf{\tilde{\Lambda}}=\mathbf{\Lambda}/\vartheta=\diag(\boldsymbol{\tilde{\eta}})+\sum_{m=1}^M\tilde{\lambda}_m\mathbf{G}_{m}^H\mathbf{G}_{m}.
\end{aligned}\end{equation}
We now consider $\boldsymbol{\tilde{\psi}}$ and $\mathbf{\tilde{Q}}_k$ as new optimization variables. Then, \eqref{eq:BD:Sumrate:dualMAC:proof:new}   can be equivalently expressed as
\begin{equation}
\begin{array}{rl}
\underset{\boldsymbol{\tilde{\psi}}\geq 0}{\mini}\;\underset{\vartheta\geq0, \{\mathbf{\tilde{Q}}_{k}\}\succeq\mathbf{0}}{\maxi} & \sum_{k=1}^K\log\frac{|\mathbf{\bar{V}}_k^H\mathbf{\tilde{\Lambda}}\mathbf{\bar{V}}_k+\tilde{\mathbf{H}}_k^H\mathbf{\tilde{Q}}_k\tilde{\mathbf{H}}_k|}{|\mathbf{\bar{V}}_k^H\mathbf{\tilde{\Lambda}}\mathbf{\bar{V}}_k|}\\
 &\qquad +\vartheta( \mathbf{p}^T\boldsymbol{\tilde{\psi}}-P)\\
\st  & \sum_{k=1}^K\tr\bigl(\mathbf{\tilde{Q}}_k\bigr)\leq  P.
\end{array}\label{eq:BD:Sumrate:dualMAC:proof:eqva}
\end{equation}
Clearly, the optimal dual variable $\vartheta^*$ can be obtained by considering the minimization of \eqref{eq:BD:Sumrate:dualMAC:proof:eqva} over $\boldsymbol{\tilde{\psi}}$. Hence, \eqref{eq:BD:Sumrate:dualMAC:proof:eqva} is the dual of the following problem:
\begin{equation}
\begin{array}{rl}
\underset{\boldsymbol{\tilde{\psi}}\geq 0}{\mini}\;\underset{\vartheta\geq0, \{\mathbf{\tilde{Q}}_{k}\}\succeq\mathbf{0}}{\maxi} & \sum_{k=1}^K\log\frac{|\mathbf{\bar{V}}_k^H\mathbf{\tilde{\Lambda}}\mathbf{\bar{V}}_k+\tilde{\mathbf{H}}_k^H\mathbf{\tilde{Q}}_k\tilde{\mathbf{H}}_k|}{|\mathbf{\bar{V}}_k^H\mathbf{\tilde{\Lambda}}\mathbf{\bar{V}}_k|}\\
\st  &  \sum_{k=1}^K\tr\bigl(\mathbf{\tilde{Q}}_k\bigr)\leq  P\\
		 & \mathbf{p}^T\boldsymbol{\tilde{\psi}}\leq P.
\end{array}\label{eq:BD:Sumrate:dualMAC:proof:dual}
\end{equation}
Finally, putting \eqref{eq:relation:SQ}, \eqref{eq:newvariable}, \eqref{eq:BD:Sumrate:dualMAC:proof:dual} and $\mathbf{S}_k=\mathbf{\Omega}_k^{-1/2}\mathbf{\bar{S}}_k\mathbf{\Omega}_k^{-1/2}$ together finalizes the proof.

 We now show how Theorem \ref{thm:BC-MAC-duality} can be modified to include a SPC constraint which is written as $\sum_{k=1}^K\mathrm{tr}(\mathbf{T}_k\mathbf{T}_k^H)\leq P$, where $P$ is the total transmit power at the secondary BS. By following the same steps from \eqref{eq:larg:expression} to \eqref{eq:BD:Sumrate:dualMAC:proof:dual}, we can arrive at the same convex-concave program as the one in Theorem  \ref{thm:BC-MAC-duality}, except that $\boldsymbol{\Lambda}$, $\boldsymbol{\psi}$, and $\mathbf{p}$ are changed to $\mathbf{\Lambda}=\eta\mathbf{I}+\sum_{m=1}^M\lambda_m\mathbf{G}_{m}^H\mathbf{G}_{m}$, $\boldsymbol{\psi}=[\eta\quad\boldsymbol{\lambda}^T]^T$, and $\mathbf{p}=[P\quad \boldsymbol{\bar{I}}^T]^T$, respectively. Consequently, Algorithm 1 can be applied to handle this case.

%


\end{document}